\newif\ifconf\conffalse

\ifconf
\documentclass[10pt,conference,compsocconf]{IEEEfocs}
\IEEEoverridecommandlockouts
\else
\documentclass[11pt,letterpaper]{article}
\fi

\usepackage{url}
\usepackage{amsthm}
\usepackage{times}
\ifconf
\else
\usepackage[margin=1in,dvips]{geometry}
\fi

\usepackage{graphicx}
\usepackage{amssymb}
\usepackage{amsmath}
\usepackage{float}
\usepackage{enumerate}

\newtheorem{theorem}{Theorem}[section]
\newtheorem{lemma}[theorem]{Lemma}

\newtheorem{definition}[theorem]{Definition}
\newtheorem{claim}[theorem]{Claim}

\newcommand{\inner}[1]{\langle#1\rangle}
\newcommand{\abs}[1]{\left|#1\right|}

\newcommand{\norm}[2]{\left \lVert#2\right \rVert_{#1}}
\newcommand{\head}[2]{#1_{[#2]}}
\newcommand{\tail}[2]{#1_{\overline{[#2]}}}
\newcommand{\range}[3]{#1_{\overline{[#2]}\cap[#3]}}

\newcommand{\gaplinf}{\mathsf{Gap}\ell_{\infty}}

\newcommand{\indlinf}{\mathsf{Ind}\ell_{\infty}}

{\makeatletter
 \gdef\xxxmark{%
   \expandafter\ifx\csname @mpargs\endcsname\relax % in minipage?
     \expandafter\ifx\csname @captype\endcsname\relax % in figure/caption?
       \marginpar{xxx}% not in a caption or minipage, can use marginpar
     \else
       xxx % notice trailing space
     \fi
   \else
     xxx % notice trailing space
   \fi}
 \gdef\xxx{\@ifnextchar[\xxx@lab\xxx@nolab}
 \long\gdef\xxx@lab[#1]#2{{\bf [\xxxmark #2 ---{\sc #1}]}}
 \long\gdef\xxx@nolab#1{{\bf [\xxxmark #1]}}
}

\DeclareMathOperator*{\argmin}{arg\,min}
\DeclareMathOperator{\supp}{supp}
\DeclareMathOperator{\E}{E}
\def\R{\mathbb{R}}
\def\eps{\epsilon}

\title{$(1+\eps)$-approximate Sparse Recovery}

\ifconf
\author{
\IEEEauthorblockN{Eric Price}
\IEEEauthorblockA{MIT CSAIL\\
ecprice@mit.edu}
\and
\IEEEauthorblockN{David P. Woodruff}
\IEEEauthorblockA{IBM Almaden\\
dpwoodru@us.ibm.com}
}
\else
\author{Eric Price\\MIT \and David P. Woodruff\\IBM Almaden}
\fi
\date{2011-08-12}

\begin{document}
\maketitle

\begin{abstract}
The problem central to sparse recovery and compressive sensing is that of \emph{stable sparse
  recovery}: we want a distribution $\mathcal{A}$ of matrices $A \in
\R^{m \times n}$ such that, for any $x \in \R^n$ and with probability
$1 - \delta > 2/3$ over $A \in \mathcal{A}$, there is an algorithm to
recover $\hat{x}$ from $Ax$ with
\begin{align}
  \norm{p}{\hat{x} - x} \leq C \min_{k\text{-sparse } x'} \norm{p}{x - x'}
\end{align}
for some constant $C > 1$ and norm $p$. 

The measurement complexity of this problem is well understood for
constant $C > 1$.  However, in a variety of applications it is
important to obtain $C = 1+\eps$ for a small $\eps > 0$, and this
complexity is not well understood.  We resolve the dependence on
$\eps$ in the number of measurements required of a $k$-sparse recovery
algorithm, up to polylogarithmic factors for the central cases of
$p=1$ and $p=2$.  Namely, we give new algorithms and lower bounds that
show the number of measurements required is $k/\eps^{p/2}
\textrm{polylog}(n)$.  For $p=2$, our bound of $\frac{1}{\eps}k\log
(n/k)$ is tight up to \emph{constant} factors. We also give matching
bounds when the output is required to be $k$-sparse, in which case we
achieve $k/\eps^p \textrm{polylog}(n)$. This shows the distinction
between the complexity of sparse and non-sparse outputs is
fundamental.
\end{abstract} 
\ifconf
\else
%\thispagestyle{empty}
%\newpage
\fi

\section{Introduction}

Over the last several years, substantial interest has been
generated in the problem of solving underdetermined linear systems
subject to a sparsity constraint.  The field, known as
\emph{compressed sensing} or \emph{sparse recovery}, has applications
to a wide variety of fields that includes data stream
algorithms~\cite{M05}, medical or geological
imaging~\cite{CRT06,D06}, and genetics testing~\cite{SAZ}. %BGSW}.
The approach uses the power of a \emph{sparsity} constraint: a vector
$x'$ is \emph{$k$-sparse} if at most $k$ coefficients are non-zero.  A
standard formulation for the problem is that of \emph{stable sparse
  recovery}: we want a distribution $\mathcal{A}$ of matrices $A \in
\R^{m \times n}$ such that, for any $x \in \R^n$ and with probability
$1 - \delta > 2/3$ over $A \in \mathcal{A}$, there is an algorithm to
recover $\hat{x}$ from $Ax$ with
\begin{align}\label{eq:lplp}
  \norm{p}{\hat{x} - x} \leq C \min_{k\text{-sparse } x'} \norm{p}{x - x'}
\end{align}
for some constant $C > 1$ and norm $p$\footnote{Some formulations
  allow the two norms to be different, in which case $C$ is not
  constant.  We only consider equal norms in this paper.}.  We call
this a \emph{$C$-approximate $\ell_p/\ell_p$ recovery scheme} with
\emph{failure probability $\delta$}.  We refer to the elements of $Ax$
as \emph{measurements}.

It is known~\cite{CRT06,GLPS} that such recovery schemes exist for $p
\in \{1,2\}$ with $C = O(1)$ and $m = O(k \log \frac{n}{k})$.
Furthermore, it is known~\cite{DIPW,FPRU} that any such recovery
scheme requires $\Omega(k \log_{1 + C} \frac{n}{k})$ measurements.
This means the measurement complexity is well understood for $C = 1 + \Omega(1)$,
but not for $C = 1 + o(1)$.

A number of applications would like to have $C = 1+\eps$ for small
$\eps$.  For example, a radio wave signal can be modeled as $x = x^* +
w$ where $x^*$ is $k$-sparse (corresponding to a signal over a narrow
band) and the noise $w$ is i.i.d. Gaussian with $\norm{p}{w} \approx D
\norm{p}{x^*}$~\cite{TDB}.  Then sparse recovery with $C = 1 +
\alpha/D$ allows the recovery of a $(1-\alpha)$ fraction of the true
signal $x^*$.  Since $x^*$ is concentrated in a small band while $w$
is located over a large region, it is often the case that $\alpha/D
\ll 1$.

The difficulty of $(1+\eps)$-approximate recovery has seemed to depend
on whether the output $x'$ is required to be $k$-sparse or can have
more than $k$ elements in its support.  Having $k$-sparse output is
important for some applications (e.g. the aforementioned radio waves)
but not for others (e.g. imaging).  Algorithms that output a
$k$-sparse $x'$ have used $\Theta(\frac{1}{\eps^p}k\log n)$
measurements~\cite{CCF,CM04,CM06,W09}.  In contrast,~\cite{GLPS} uses
only $\Theta(\frac{1}{\eps}k\log (n/k))$ measurements for $p = 2$ and
outputs a non-$k$-sparse $x'$.

\begin{figure*}\label{fig:results}
  \begin{center}
    \renewcommand{\arraystretch}{1.5}
    \begin{tabular}{|l|l|l|l|}
      \hline
      &  & Lower bound & Upper bound\\
      \hline
      $k$-sparse output  & $\ell_1$ & $\Omega(\frac{1}{\eps}(k \log \frac{1}{\eps} + \log \frac{1}{\delta}))$ & $O(\frac{1}{\eps} k \log n)$\cite{CM04}\\
      & $\ell_2$ & $\Omega(\frac{1}{\eps^2}(k + \log \frac{1}{\delta}))$ & $O(\frac{1}{\eps^2} k \log n)$\cite{CCF,CM06,W09}\\
      \hline
      Non-$k$-sparse output  & $\ell_1$ & $\Omega(\frac{1}{\sqrt{\eps}\log^2(k/\eps)}k)$ & $O(\frac{\log^3(1/\eps)}{\sqrt{\eps}}k\log n)$\\
      & $\ell_2$ & $\Omega(\frac{1}{\eps}k \log (n/k))$ & $O(\frac{1}{\eps}k \log (n/k))$\cite{GLPS}\\
      \hline
    \end{tabular}
    \renewcommand{\arraystretch}{1}
  \end{center}
  \caption{Our results, along with existing upper bounds.  Fairly minor
    restrictions on the relative magnitude of parameters apply; see
    the theorem statements for details.}
\end{figure*}

\paragraph{Our results}
We show that the apparent distinction between complexity of sparse and
non-sparse outputs is fundamental, for both $p = 1$ and $p = 2$.  We
show that for sparse output, $\Omega(k / \eps^p)$ measurements are
necessary, matching the upper bounds up to a $\log n$ factor.  For
general output and $p=2$, we show $\Omega(\frac{1}{\eps}k \log (n/k))$
measurements are necessary, matching the upper bound up to a constant
factor.  In the remaining case of general output and $p=1$, we show
$\widetilde\Omega(k / \sqrt{\eps})$ measurements are necessary.  We
then give a novel algorithm that uses
$O(\frac{\log^3(1/\eps)}{\sqrt{\eps}}k \log n)$ measurements, beating
the $1/\eps$ dependence given by all previous algorithms.  As a
result, all our bounds are tight up to factors logarithmic in $n$.
The full results are shown in Figure~\ref{fig:results}.

In addition, for $p = 2$ and general output, we show that thresholding
the top $2k$ elements of a Count-Sketch~\cite{CCF} estimate gives
$(1+\eps)$-approximate recovery with $\Theta(\frac{1}{\eps}k\log n)$
measurements.  This is interesting because it highlights the
distinction between sparse output and non-sparse output:~\cite{CM06}
showed that thresholding the top $k$ elements of a Count-Sketch
estimate requires $m = \Theta(\frac{1}{\eps^2}k\log n)$.
While~\cite{GLPS} achieves $m = \Theta(\frac{1}{\eps}k\log (n/k))$ for
the same regime, it only succeeds with constant probability while ours
succeeds with probability $1 - n^{-\Omega(1)}$; hence ours is the most
efficient known algorithm when $\delta = o(1), \eps = o(1),$ and $k <
n^{0.9}$.

\paragraph{Related work}
Much of the work on sparse recovery has relied on the Restricted
Isometry Property~\cite{CRT06}.  None of this work has been able to
get better than $2$-approximate recovery, so there are relatively few
papers achieving $(1+\eps)$-approximate recovery.  The existing ones
with $O(k \log n)$ measurements are surveyed above (except
for~\cite{IR}, which has worse dependence on $\eps$ than~\cite{CM04}
for the same regime).

A couple of previous works have studied the $\ell_\infty/\ell_p$
problem, where every coordinate must be estimated with small error.
This problem is harder than $\ell_p/\ell_p$ sparse recovery with
sparse output.  For $p=2$,~\cite{W09} showed that schemes using
Gaussian matrices $A$ require $m = \Omega(\frac{1}{\eps^2} k \log
(n/k))$.  For $p=1$,~\cite{CM05} showed that any sketch requires
$\Omega(k/\eps)$ bits (rather than measurements).

Independently of this work and of each other, multiple
authors~\cite{CD11,IT10,ASZ10} have matched our
$\Omega(\frac{1}{\eps}k \log (n/k))$ bound for $\ell_2/\ell_2$ in
related settings.  The details vary, but all proofs are broadly
similar in structure to ours: they consider observing a large set of
``well-separated'' vectors under Gaussian noise.  Fano's inequality
gives a lower bound on the mutual information between the observation
and the signal; then, an upper bound on the mutual information is
given by either the Shannon-Hartley theorem or a KL-divergence
argument.  This technique does not seem useful for the other problems
we consider in this paper, such as lower bounds for $\ell_1/\ell_1$ or
the sparse output setting.

\paragraph{Our techniques}
For the upper bounds for non-sparse output, we observe that the hard
case for sparse output is when the noise is fairly concentrated, in
which the estimation of the top $k$ elements can have $\sqrt{\eps}$
error.  Our goal is to recover enough mass from outside the top $k$
elements to cancel this error.  The upper bound for $p = 2$ is a
fairly straightforward analysis of the top $2k$ elements of a
Count-Sketch data structure.

The upper bound for $p = 1$ proceeds by subsampling the vector at rate
$2^{-i}$ and performing a Count-Sketch with size proportional to
$\frac{1}{\sqrt{\eps}}$, for $i \in \{0,1, \dotsc, O(\log
(1/\eps))\}$.  The intuition is that if the noise is well spread over
many (more than $k/\eps^{3/2}$) coordinates, then the $\ell_2$ bound
from the first Count-Sketch gives a very good $\ell_1$ bound, so the
approximation is $(1+\eps)$-approximate.  However, if the noise
is concentrated over a small number $k/\eps^c$ of coordinates, then
the error from the first Count-Sketch is proportional to $1 +
\eps^{c/2 + 1/4}$.  But in this case, one of the subsamples will only
have $O(k/\eps^{c/2 - 1/4}) < k/\sqrt{\eps}$ of the coordinates with
large noise.  We can then recover those coordinates with the
Count-Sketch for that subsample.  Those coordinates contain an
$\eps^{c/2 + 1/4}$ fraction of the total noise, so recovering them
decreases the approximation error by exactly the error induced from
the first Count-Sketch.

The lower bounds use substantially different techniques for sparse
output and for non-sparse output.  For sparse output, we use
reductions from communication complexity to show a lower bound in
terms of bits.  Then, as in~\cite{DIPW}, we embed $\Theta(\log n)$
copies of this communication problem into a single vector.  This
multiplies the bit complexity by $\log n$; we also show we can round
$Ax$ to $\log n$ bits per measurement without affecting recovery,
giving a lower bound in terms of measurements.

We illustrate the lower bound on bit complexity for sparse output
using $k=1$.  Consider a vector $x$ containing $1/\eps^p$ ones and
zeros elsewhere, such that $x_{2i} + x_{2i+1} = 1$ for all $i$.  For
any $i$, set $z_{2i} = z_{2i+1} = 1$ and $z_{j} = 0$ elsewhere.  Then
successful $(1+\eps/3)$-approximate sparse recovery from $A(x + z)$
returns $\hat{z}$ with $\supp(\hat{z}) = \supp(x) \cap \{2i, 2i+1\}$.
Hence we can recover each bit of $x$ with probability $1-\delta$,
requiring $\Omega(1 / \eps^p)$ bits\footnote{For $p = 1$, we can
  actually set $\abs{\supp(z)} = 1/\eps$ and search among a set of
  $1/\eps$ candidates.  This gives $\Omega(\frac{1}{\eps}\log
  (1/\eps))$ bits.}.  We can generalize this to $k$-sparse output for
$\Omega(k / \eps^p)$ bits, and to $\delta$ failure probability with
$\Omega(\frac{1}{\eps^p}\log \frac{1}{\delta})$.  However, the two
generalizations do not seem to combine.

For non-sparse output, we split between $\ell_2$ and $\ell_1$.  In
$\ell_2$, we consider $A(x + w)$ where $x$ is sparse and $w$ has
uniform Gaussian noise with $\norm{2}{w}^2 \approx
\norm{2}{x}^2/\eps$.  Then each coordinate of $y = A(x + w) = Ax + Aw$
is a Gaussian channel with signal to noise ratio $\eps$.  This channel
has channel capacity $\eps$, showing $I(y; x) \leq \eps m$.  Correct
sparse recovery must either get most of $x$ or an $\eps$ fraction of
$w$; the latter requires $m = \Omega(\eps n)$ and the former requires
$I(y; x) = \Omega(k \log (n/k))$.  This gives a tight
$\Theta(\frac{1}{\eps}k \log (n/k))$ result.  Unfortunately, this does
not easily extend to $\ell_1$, because it relies on the Gaussian
distribution being both stable and maximum entropy under $\ell_2$; the
corresponding distributions in $\ell_1$ are not the same.

Therefore for $\ell_1$ non-sparse output, we have yet another
argument.  The hard instances for $k=1$ must have one large value (or
else $0$ is a valid output) but small other values (or else the
$2$-sparse approximation is significantly better than the $1$-sparse
approximation).  Suppose $x$ has one value of size $\eps$ and $d$
values of size $1/d$ spread through a vector of size $d^2$.  Then a
$(1 + \eps/2)$-approximate recovery scheme must either locate the
large element or guess the locations of the $d$ values with
$\Omega(\eps d)$ more correct than incorrect.  The former requires $1
/ (d\eps^2)$ bits by the difficulty of a novel version of the
Gap-$\ell_\infty$ problem.  The latter requires $\eps d$ bits because
it allows recovering an error correcting code.  Setting $d =
\eps^{-3/2}$ balances the terms at $\eps^{-1/2}$ bits.  Because some
of these reductions are very intricate, this extended abstract does
not manage to embed $\log n$ copies of the problem into a single
vector.  As a result, we lose a $\log n$ factor in a universe of size
$n = \text{poly}(k/\eps)$ when converting to measurement complexity
from bit complexity.

\section{Preliminaries}

\paragraph{Notation}
We use $[n]$ to denote the set $\{1 \ldots n\}$.  For any set $S
\subset [n]$, we use $\overline{S}$ to denote the complement of $S$,
i.e., the set $[n]\setminus S$.  For any $x \in \R^n$, $x_i$ denotes
the $i$th coordinate of $x$, and $x_S$ denotes the vector $x' \in
\R^n$ given by $x'_i = x_i$ if $i \in S$, and $x'_i = 0$ otherwise.
We use $\supp(x)$ to denote the support of $x$.

\section{Upper bounds}

The algorithms in this section are indifferent to permutation of the
coordinates.  Therefore, for simplicity of notation in the analysis,
we assume the coefficients of $x$ are sorted such that $\abs{x_1} \geq
\abs{x_2} \geq \dotsc \geq \abs{x_n} \geq 0$.

\paragraph{Count-Sketch}
Both our upper bounds use the Count-Sketch~\cite{CCF} data structure.
The structure consists of $c\log n$ hash tables of size $O(q)$, for
$O(cq\log n)$ total space; it can be represented as $Ax$ for a matrix
$A$ with $O(cq\log n)$ rows.  Given $Ax$, one can construct $x^*$
with
\begin{align}\label{eq:count-sketch}
  \norm{\infty}{x^* - x}^2 \leq \frac{1}{q}\norm{2}{x_{\overline{[q]}}}^2
\end{align}
with failure probability $n^{1-c}$.

\subsection{Non-sparse $\ell_2$}

It was shown in~\cite{CM06} that, if $x^*$ is the result of a
Count-Sketch with hash table size $O(k/\eps^2)$, then outputting the
top $k$ elements of $x^*$ gives a $(1+\eps)$-approximate
$\ell_2/\ell_2$ recovery scheme.  Here we show that a seemingly minor
change---selecting $2k$ elements rather than $k$ elements---turns this
into a $(1+\eps^2)$-approximate $\ell_2/\ell_2$ recovery scheme.

\begin{theorem}\label{thm:csExtended}
  Let $\hat{x}$ be the top $2k$ estimates from a Count-Sketch
  structure with hash table size $O(k/\eps)$.  Then with failure
  probability $n^{-\Omega(1)}$,
  \[
  \norm{2}{\hat{x}-x} \leq (1 + \eps)\norm{2}{\tail{x}{k}}.
  \]
  Therefore, there is a $1+\eps$-approximate $\ell_2/\ell_2$ recovery
  scheme with $O(\frac{1}{\eps}k\log n)$ rows.
\end{theorem}
\begin{proof}
  Let the hash table size be $O(ck/\eps)$ for constant $c$, and let
  $x^*$ be the vector of estimates for each coordinate.  Define $S$ to
  be the indices of the largest $2k$ values in $x^*$, and $E =
  \norm{2}{\tail{x}{k}}$.

  By~\eqref{eq:count-sketch}, the standard analysis of Count-Sketch:
  \[
  \norm{\infty}{x^*-x}^2 \leq \frac{\eps}{ck}E^2.
  \]
  so
  \begin{align}
%  \ifconf\notag&\fi %XXX remove in non-conf
    \norm{2}{x^*_S - x}^2 - E^2 
  \ifconf\\\notag\fi
  =
%  \ifconf&\fi %XXX remove in non-conf
  \norm{2}{x^*_S - x}^2 - \norm{2}{\tail{x}{k}}^2
    \ifconf\\\fi
    \leq& \norm{2}{(x^*-x)_S}^2 +
    \norm{2}{x_{[n]\setminus S}}^2 - \norm{2}{\tail{x}{k}}^2\notag\\
    \leq& \abs{S} \norm{\infty}{x^*-x}^2 + \norm{2}{x_{[k] \setminus S}}^2 - \norm{2}{x_{S \setminus [k]}}^2\notag\\
    \leq& \frac{2\eps}{c} E^2 + \norm{2}{x_{[k] \setminus S}}^2 - \norm{2}{x_{S \setminus [k]}}^2\label{eq:l2l2initial}
  \end{align}

  Let $a = \max_{i \in [k]\setminus S} x_i$ and $b = \min_{i \in S
    \setminus [k]} x_i$, and let $d = \abs{[k] \setminus S}$.  The
  algorithm passes over an element of value $a$ to choose one of value
  $b$, so
  \[
  a \leq b + 2\norm{\infty}{x^*-x} \leq b +  2\sqrt{\frac{\eps}{ck}}E.
  \]
  Then
  \begin{align*}
%  \ifconf\notag&\fi %XXX remove in non-conf
    \norm{2}{x_{[k] \setminus S}}^2 - \norm{2}{x_{S \setminus [k]}}^2  
    \ifconf\\\fi
    \leq& da^2 - (k + d)b^2\\
    \leq& d(b + 2\sqrt{\frac{\eps}{ck}}E)^2 - (k+d)b^2\\
    \leq& -kb^2 + 4\sqrt{\frac{\eps}{ck}}dbE + \frac{4\eps}{ck}dE^2\\
    \leq& -k(b - 2\sqrt{\frac{\eps }{ck^3}} dE)^2 + \frac{4\eps}{ck^2}dE^2(k-d)\\
    \leq& \frac{4d(k-d)\eps}{ck^2}E^2 \leq \frac{\eps}{c}E^2
  \end{align*}
  and combining this with~\eqref{eq:l2l2initial} gives
  \[
  \norm{2}{x^*_S - x}^2 - E^2\leq \frac{3\eps}{c}E^2
  \]
  or
  \[
  \norm{2}{x^*_S - x} \leq (1+\frac{3\eps}{2c})E
  \]
  which proves the theorem for $c \geq 3/2$.
\end{proof}

\subsection{Non-sparse $\ell_1$}

\begin{theorem}\label{thm:l1upper}
  There exists a $(1+\eps)$-approximate $\ell_1/\ell_1$ recovery
  scheme with $O(\frac{\log^3 1/\eps}{\sqrt{\eps}} k\log n)$
  measurements and failure probability $e^{-\Omega(k/\sqrt{\eps})} +
  n^{-\Omega(1)}$.
\end{theorem}

Set $f = \sqrt{\eps}$, so our goal is to get $(1+f^2)$-approximate
$\ell_1/\ell_1$ recovery with $O(\frac{\log^3 1/f}{f}k\log n)$
measurements.

For intuition, consider 1-sparse recovery of the following vector $x$:
let $c \in [0, 2]$ and set $x_1 = 1/f^9$ and $x_2, \dotsc,
x_{1 + 1/f^{1+c}} \in \{\pm 1\}$.  Then we have
\begin{align*}
  \norm{1}{\tail{x}{1}} &= 1/f^{1+c}
\end{align*}
and by~\eqref{eq:count-sketch}, a Count-Sketch with $O(1/f)$-sized
hash tables returns $x^*$ with
\begin{align*}
  \norm{\infty}{x^*-x} \leq \sqrt{f}\norm{2}{\tail{x}{1/f}} \approx 1/f^{c/2} = f^{1 + c/2}\norm{1}{\tail{x}{1}}.
\end{align*}
The reconstruction algorithm therefore cannot reliably find any of the
$x_i$ for $i > 1$, and its error on $x_1$ is at least $f^{1 +
  c/2}\norm{1}{\tail{x}{1}}$.  Hence the algorithm will not do better
than a $f^{1 + c/2}$-approximation.

However, consider what happens if we subsample an $f^c$ fraction of
the vector.  The result probably has about $1/f$ non-zero values, so a
$O(1/f)$-width Count-Sketch can reconstruct it exactly.  Putting this
in our output improves the overall $\ell_1$ error by about $1/f =
f^c\norm{1}{\tail{x}{1}}$.  Since $c < 2$, this more than cancels the
$f^{1+c/2}\norm{1}{\tail{x}{1}}$ error the initial Count-Sketch makes
on $x_1$, giving an approximation factor better than $1$.

This tells us that subsampling can help.  We don't need to subsample
at a scale below $k/f$ (where we can reconstruct well already) or
above $k/f^3$ (where the $\ell_2$ bound is small enough already), but
in the intermediate range we need to subsample.  Our algorithm
subsamples at all $\log 1/f^2$ rates in between these two endpoints,
and combines the heavy hitters from each.

First we analyze how subsampled Count-Sketch works.

\begin{lemma}\label{lemma:subsampledsketch}
  Suppose we subsample with probability $p$ and then apply
  Count-Sketch with $\Theta(\log n)$ rows and $\Theta(q)$-sized hash
  tables.  Let $y$ be the subsample of $x$.  Then with failure
  probability $e^{-\Omega(q)} + n^{-\Omega(1)}$ we recover a $y^*$ with
  \[
  \norm{\infty}{y^* - y} \leq \sqrt{p/q}\norm{2}{\tail{x}{q/p}}.
  \]
\end{lemma}
\begin{proof}
  Recall the following form of the Chernoff bound: if $X_1, \dotsc,
  X_m$ are independent with $0 \leq X_i \leq M$, and $\mu \geq \E[\sum
    X_i]$, then
  \[
  \Pr[\sum X_i \geq \frac{4}{3}\mu] \leq e^{-\Omega(\mu/M)}.
  \]

  Let $T$ be the set of coordinates in the sample.  Then $\E[\abs{T
    \cap [\frac{3q}{2p}]}] = 3q/2$, so
  \[
  \Pr\left[ \abs{T \cap [\frac{3q}{2p}]} \geq 2q\right] \leq e^{-\Omega(q)}.
  \]
  Suppose this event does not happen, so $\abs{T \cap [\frac{3q}{2p}]}
  < 2q$.  We also have
  \[
  \norm{2}{\tail{x}{q/p}} \geq \sqrt{\frac{q}{2p}}\abs{x_{\frac{3q}{2p}}}.
  \]
  Let $Y_i = 0$ if $i \notin T$ and $Y_i = x_i^2$ if $i \in T$.  Then
  \[
  \E[\sum_{i > \frac{3q}{2p}} Y_i] = p\norm{2}{\tail{x}{\frac{3q}{2p}}}^2 \leq p\norm{2}{\tail{x}{q/p}}^2
  \]
  For $i > \frac{3q}{2p}$ we have
  \[
  Y_i \leq \abs{x_{\frac{3q}{2p}}}^2 \leq \frac{2p}{q}\norm{2}{\tail{x}{q/p}}^2
  \]
  giving by Chernoff that 
  \begin{align*}
    \Pr[\sum Y_i \geq \frac{4}{3}p\norm{2}{\tail{x}{q/p}}^2] \leq e^{-\Omega(q/2)}
  \end{align*}
  But if this event does not happen, then
  \begin{align*}
    \norm{2}{\tail{y}{2q}}^2 \leq \sum_{i \in T, i > \frac{3q}{2p}} x_i^2 = \sum_{i > \frac{3q}{2p}} Y_i \leq \frac{4}{3}p\norm{2}{\tail{x}{q/p}}^2
  \end{align*}
  By~\eqref{eq:count-sketch}, using $O(2q)$-size hash tables gives a
  $y^*$ with
  \[
  \norm{\infty}{y^* - y} \leq \frac{1}{\sqrt{2q}} \norm{2}{\tail{y}{2q}} \leq \sqrt{p/q}\norm{2}{\tail{x}{q/p}}
  \]
  with failure probability $n^{-\Omega(1)}$, as desired.
\end{proof}

Let $r = 2\log 1/f$.  Our algorithm is as follows: for $j \in
\{0,\dotsc, r\}$, we find and estimate the $2^{j/2}k$ largest elements
not found in previous $j$ in a subsampled Count-Sketch with
probability $p = 2^{-j}$ and hash size $q = ck/f$
for some parameter $c = \Theta(r^2)$.  We output $\hat{x}$, the union of all
these estimates.  Our goal is to show
\begin{align*}
  \norm{1}{\hat{x}-x} - \norm{1}{\tail{x}{k}} \leq O(f^2)\norm{1}{\tail{x}{k}}.
\end{align*}

For each level $j$, let $S_j$ be the $2^{j/2}k$ largest coordinates in
our estimate not found in $S_1 \cup \dotsb \cup S_{j-1}$.  Let $S =
\cup S_j$.  By Lemma~\ref{lemma:subsampledsketch}, for each $j$ we
have (with failure probability $e^{-\Omega(k/f)} + n^{-\Omega(1)}$) that
\begin{align*}
  \norm{1}{(\hat{x}-x)_{S_j}} &\leq \abs{S_j}\sqrt{\frac{2^{-j} f}{ck}} \norm{2}{\tail{x}{2^jck/f}}\\
  &\leq 2^{-j/2}\sqrt{\frac{fk}{c}} \norm{2}{\tail{x}{2k/f}}
\end{align*}
and so
\begin{align}\label{eq:xSerror}
  \norm{1}{(\hat{x}-x)_{S}} &= \sum_{j=0}^r \norm{1}{(\hat{x}-x)_{S_j}}
%\ifconf\notag\\&\fi %XXX remove in non-conf
  \leq \frac{1}{(1-1/\sqrt{2})\sqrt{c}} \sqrt{fk} \norm{2}{\tail{x}{2k/f}}
\end{align}

By standard arguments, the $\ell_\infty$ bound for $S_0$ gives
\begin{align}\label{eq:ellinfell12}
  \norm{1}{\head{x}{k}} \leq \norm{1}{x_{S_0}} + k\norm{\infty}{\hat{x}_{S_0} -x_{S_0}} \leq \sqrt{fk/c}\norm{2}{\tail{x}{2k/f}}
\end{align}

Combining Equations~\eqref{eq:xSerror} and~\eqref{eq:ellinfell12} gives
\begin{align}
%  \ifconf&\fi %XXX remove in non-conf
  \norm{1}{\hat{x} - x} - \norm{1}{\tail{x}{k}}
  \ifconf\\\fi
  =& \norm{1}{(\hat{x}-x)_{S}} + \norm{1}{x_{\overline{S}}} - \norm{1}{\tail{x}{k}}\notag\\
  =& \norm{1}{(\hat{x}-x)_{S}} + \norm{1}{\head{x}{k}} - \norm{1}{x_{S}}\notag\\
  =& \norm{1}{(\hat{x}-x)_{S}} + (\norm{1}{\head{x}{k}} - \norm{1}{x_{S_0}}) - \sum_{j=1}^r \norm{1}{x_{S_j}}\notag\\
  \leq& \left(\frac{1}{(1-1/\sqrt{2})\sqrt{c}} + \frac{1}{\sqrt{c}}\right)\sqrt{fk}\norm{2}{\tail{x}{2k/f}}
%  \ifconf\notag\\&\quad\fi %XXX remove in non-conf
  - \sum_{j=1}^r \norm{1}{x_{S_j}}\notag\\
  =& O(\frac{1}{\sqrt{c}})\sqrt{fk}\norm{2}{\tail{x}{2k/f}} - \sum_{j=1}^r \norm{1}{x_{S_j}}\label{eq:totalerror}
\end{align}

We would like to convert the first term to depend on the $\ell_1$
norm.  For any $u$ and $s$ we have, by splitting into chunks of size
$s$, that
\begin{align*}
  \norm{2}{\tail{u}{2s}} \leq \sqrt{\frac{1}{s}}\norm{1}{\tail{u}{s}}\\
  \norm{2}{\range{u}{s}{2s}} \leq \sqrt{s}\abs{u_{s}}.
\end{align*}
Along with the triangle inequality, this gives us that
\begin{align}
  \sqrt{kf}\norm{2}{\tail{x}{2k/f}} &\leq \sqrt{kf}\norm{2}{\tail{x}{2k/f^3}}
%  \ifconf\notag\\&\qquad\fi %XXX remove in non-conf
 + \sqrt{kf}\sum_{j=1}^{r} \norm{2}{\range{x}{2^jk/f}{2^{j+1}k/f}}\notag\\
  &\leq f^2\norm{1}{\tail{x}{k/f^3}} + \sum_{j=1}^{r} k2^{j/2}\abs{x_{2^jk/f}}\notag
\end{align}
so
\begin{align}
%  \ifconf\notag&\fi %XXX remove in non-conf
  \norm{1}{\hat{x} - x} - \norm{1}{\tail{x}{k}}
\ifconf\\\fi
 \leq&  O(\frac{1}{\sqrt{c}})f^2\norm{1}{\tail{x}{k/f^3}} + \sum_{j=1}^{r} O(\frac{1}{\sqrt{c}})k2^{j/2}\abs{x_{2^jk/f}}
%  \ifconf\notag\\&\quad\fi %XXX remove in non-conf
- \sum_{j=1}^{r}\norm{1}{x_{S_j}} \label{eq:l2l1telescope}
\end{align}
Define $a_j = k2^{j/2}\abs{x_{2^jk/f}}$.  The first term grows as
$f^2$ so it is fine, but $a_j$ can grow as $f2^{j/2} > f^2$.  We need to
show that they are canceled by the corresponding $\norm{1}{x_{S_j}}$.
In particular, we will show that $\norm{1}{x_{S_j}} \geq \Omega(a_j) -
O(2^{-j/2}f^2 \norm{1}{\tail{x}{k/f^3}})$ with high probability---at
least wherever $a_j \geq \norm{1}{a}/(2r)$.

Let $U \in [r]$ be the set of $j$ with $a_j \geq \norm{1}{a} / (2r)$,
so that $\norm{1}{a_{U}} \geq \norm{1}{a}/2$.  We have

\begin{align}
  \norm{2}{\tail{x}{2^jk/f}}^2 &= \norm{2}{\tail{x}{2k/f^3}}^2 +
  \sum_{i=j}^r \norm{2}{\range{x}{2^jk/f}{2^{j+1}k/f}}^2\notag\\
  &\leq \norm{2}{\tail{x}{2k/f^3}}^2 + \frac{1}{kf}\sum_{i=j}^r a_j^2\label{eq:taill2squared}
\end{align}

For $j \in U$, we have
\begin{align*}
  \sum_{i=j}^r a_i^2 \leq a_j \norm{1}{a} \leq 2ra_j^2
\end{align*}
so, along with $(y^2+z^2)^{1/2} \leq y+z$, we turn Equation~\eqref{eq:taill2squared} into
\begin{align*}
  \norm{2}{\tail{x}{2^jk/f}} &\leq \norm{2}{\tail{x}{2k/f^3}} +
  \sqrt{\frac{1}{kf}\sum_{i=j}^r a_j^2}\\
  &\leq \sqrt{\frac{f^3}{k}}\norm{1}{\tail{x}{k/f^3}} + \sqrt{\frac{2r}{kf}}a_j
\end{align*}

When choosing $S_j$, let $T \in [n]$ be the set of indices chosen in
the sample.  Applying Lemma~\ref{lemma:subsampledsketch} the estimate
$x^*$ of $x_T$ has
\begin{align*}
  \norm{\infty}{x^* - x_T} &\leq \sqrt{\frac{f}{2^jck}}\norm{2}{\tail{x}{2^jk/f}}\\
  &\leq \sqrt{\frac{1}{2^jc}}\frac{f^2}{k}\norm{1}{\tail{x}{k/f^3}} + \sqrt{\frac{2r}{2^jc}}\frac{a_j}{k}\\
  &= \sqrt{\frac{1}{2^jc}}\frac{f^2}{k}\norm{1}{\tail{x}{k/f^3}} + \sqrt{\frac{2r}{c}}\abs{x_{2^jk/f}}
\end{align*}
for $j \in U$.

Let $Q = [2^jk/f] \setminus (S_0 \cup \dotsb \cup S_{j-1})$.  We have
$\abs{Q} \geq 2^{j-1}k/f$ so $\E[\abs{Q \cap T}] \geq k/2f$ and
$\abs{Q \cap T} \geq k/4f$ with failure probability
$e^{-\Omega(k/f)}$.  Conditioned on $\abs{Q \cap T} \geq k/4f$, since
$x_{T}$ has at least $\abs{Q \cap T} \geq k/(4f) = 2^{r/2}k/4 \geq
2^{j/2}k/4$ possible choices of value at least $\abs{x_{2^jk/f}}$,
$x_{S_j}$ must have at least $k2^{j/2}/4$ elements at least
$\abs{x_{2^jk/f}} - \norm{\infty}{x^*-x_T}$.  Therefore, for $j \in U$,
\begin{align*}
  \norm{1}{x_{S_j}} &\geq - \frac{1}{4\sqrt{c}} f^2 \norm{1}{\tail{x}{k/f^3}} + \frac{k2^{j/2}}{4}(1 - \sqrt{\frac{2r}{c}})\abs{x_{2^jk/f}}
\end{align*}
and therefore
\begin{align}\label{eq:subsamples}
%  \ifconf\notag&\fi %XXX remove in non-conf
  \sum_{j=1}^r \norm{1}{x_{S_j}} \geq \sum_{j \in U} \norm{1}{x_{S_j}}
%  \ifconf\notag\\\fi %XXX remove in non-conf
 \geq& \sum_{j \in U} - \frac{1}{4\sqrt{c}} f^2 \norm{1}{\tail{x}{k/f^3}} + \frac{k2^{j/2}}{4}(1 - \sqrt{\frac{2r}{c}})\abs{x_{2^jk/f}} \notag\\
  \geq& -\frac{r}{4\sqrt{c}}f^2\norm{1}{\tail{x}{k/f^3}} + \frac{1}{4}(1 - \sqrt{\frac{2r}{c}})\norm{1}{a_U}\notag\\
  \geq& -\frac{r}{4\sqrt{c}}f^2\norm{1}{\tail{x}{k/f^3}} + \frac{1}{8}(1 - \sqrt{\frac{2r}{c}}) \sum_{j=1}^r k2^{j/2}\abs{x_{2^jk/f}}
\end{align}

Using~\eqref{eq:l2l1telescope} and~\eqref{eq:subsamples} we get
\begin{align*}
%  \ifconf\notag&\fi %XXX remove in non-conf
  \norm{1}{\hat{x} - x} - \norm{1}{\tail{x}{k}}
%  \ifconf\notag\\\fi %XXX remove in non-conf
\leq&  \left(\frac{r}{4\sqrt{c}} + O(\frac{1}{\sqrt{c}})\right)f^2\norm{1}{\tail{x}{k/f^3}}
%  \ifconf\notag\\&\quad\fi %XXX remove in non-conf
 + \sum_{j=1}^{r} \left(O(\frac{1}{\sqrt{c}}) + \frac{1}{8}\sqrt{\frac{2r}{c}}-\frac{1}{8}\right)k2^{j/2}\abs{x_{2^jk/f}}\\
  \leq& f^2\norm{1}{\tail{x}{k/f^3}} \leq f^2\norm{1}{\tail{x}{k}}
\end{align*}
for some $c = O(r^2)$.  Hence we use a total of $\frac{rc}{f} k\log n
= \frac{\log^3 1/f}{f}k\log n$ measurements for $1+f^2$-approximate
$\ell_1/\ell_1$ recovery.

For each $j \in \{0,\dotsc, r\}$ we had failure probability
$e^{-\Omega(k/f)} + n^{-\Omega(1)}$ (from Lemma~\ref{lemma:subsampledsketch}
and $\abs{Q \cap T} \geq k/2f$).  By the union bound, our overall
failure probability is at most
\[
(\log \frac{1}{f})(e^{-\Omega(k/f)} + n^{-\Omega(1)}) \leq e^{-\Omega(k/f)} + n^{-\Omega(1)},
\]
proving Theorem~\ref{thm:l1upper}.

\section{Lower bounds for non-sparse output and $p=2$}

In this case, the lower bound follows fairly straightforwardly from the
Shannon-Hartley information capacity of a Gaussian channel.

We will set up a communication game.  Let
$\mathcal{F} \subset \{S \subset [n] \mid \abs{S} = k\}$ be a family
of $k$-sparse supports such that:
\begin{itemize}
\item $\abs{S \Delta S'} \geq k$ for $S \neq S' \in \mathcal{F}$,
\item $\Pr_{S \in \mathcal{F}} [i \in S] = k/n$ for all $i \in [n]$, and
\item $\log \abs{\mathcal{F}} = \Omega(k \log (n/k))$.
\end{itemize}
This is possible; for example, a random linear code on $[n/k]^k$ with
relative distance $1/2$ has these
properties~\cite{Gnotes}.\footnote{This assumes $n/k$ is a prime power
  larger than 2.  If $n/k$ is not prime, we can choose $n' \in [n/2,
  n]$ to be a prime multiple of $k$, and restrict to the first $n'$
  coordinates.  This works unless $n/k < 3$, in which case a bound of
  $\Theta(\min(n, \frac{1}{\eps}k \log (n/k))) = \Theta(k)$ is
  trivial.}

Let $X = \{x \in \{0, \pm 1\}^n \mid \supp(x) \in \mathcal{F}\}$.  Let
$w \sim N(0, \alpha\frac{k}{n}I_n)$ be i.i.d. normal with variance
$\alpha k/n$ in each coordinate.  Consider the following process:

\paragraph{Procedure} First, Alice chooses $S \in \mathcal{F}$
uniformly at random, then $x \in X$ uniformly at random subject to
$\supp(x) = S$, then $w\sim N(0, \alpha\frac{k}{n}I_n)$.  She sets $y
= A(x + w)$ and sends $y$ to Bob. Bob performs sparse recovery on $y$
to recover $x' \approx x$, rounds to $X$ by $\hat{x} =
\argmin_{\hat{x} \in X} \norm{2}{\hat{x} - x'}$, and sets $S' =
\supp(\hat{x})$.  This gives a Markov chain $S \to x \to y \to x' \to
S'$.

If sparse recovery works for any $x + w$ with probability $1-\delta$
as a distribution over $A$, then there is some specific $A$ and random
seed such that sparse recovery works with probability $1-\delta$ over
$x + w$; let us choose this $A$ and the random seed, so that Alice and
Bob run deterministic algorithms on their inputs.

\begin{lemma}\label{lemma:infupper}
  $I(S; S') = O(m \log (1 + \frac{1}{\alpha}))$.
\end{lemma}
\begin{proof}
  Let the columns of $A^T$ be $v^1, \dotsc, v^m$.  We may assume that
  the $v^i$ are orthonormal, because this can be accomplished via a
  unitary transformation on $Ax$.  Then we have that $y_i =
  \inner{v^i, x + w} = \inner{v^i, x} + w'_i$, where $w'_i \sim N(0,
  \alpha k\norm{2}{v^i}^2/n) = N(0, \alpha k/n)$ and
  \[
  \E_x[\inner{v^i, x}^2] = \E_{S}[ \sum_{j \in S} (v^i_j)^2] =
  \frac{k}{n}
  \]
  Hence $y_i = z_i + w'_i$ is a Gaussian channel with power constraint
  $\E[z_i^2] \leq \frac{k}{n}\norm{2}{v^i}^2$ and noise variance
  $\E[(w'_i)^2] = \alpha \frac{k}{n}\norm{2}{v^i}^2$.  Hence by the
  Shannon-Hartley theorem this channel has information capacity
  \[
  \max_{v_i} I(z_i; y_i) = C \leq \frac{1}{2} \log (1 +
  \frac{1}{\alpha}).
  \]
  By the data processing inequality for Markov chains and the chain
  rule for entropy, this means
  \begin{align}
    I(S; S') &\leq I(z; y) = H(y) - H(y \mid z) = H(y) - H(y-z \mid z)
    \notag\\&= H(y) - \sum H(w'_i \mid z, w'_1, \dotsc, w'_{i-1}) 
    \notag\\&= H(y) - \sum H(w'_i) \notag \leq \sum H(y_i) - H(w'_i)
    \notag\\&= \sum H(y_i) - H(y_i \mid z_i) =
    \sum I(y_i; z_i)
    \notag\\&\leq \frac{m}{2} \log (1 + \frac{1}{\alpha}).
  \end{align}
\end{proof}

We will show that successful recovery either recovers most of $x$, in
which case $I(S; S') = \Omega(k \log (n/k))$, or recovers an $\eps$
fraction of $w$.  First we show that recovering $w$ requires $m =
\Omega(\eps n)$.

\begin{lemma}\label{lemma:gaussianhard}
  Suppose $w \in \R^n$ with $w_i \sim N(0, \sigma^2)$ for all $i$ and
  $n = \Omega(\frac{1}{\eps^2}\log (1/\delta))$, and
  $A \in \R^{m \times n}$ for $m < \delta\eps n$.  Then any algorithm
  that finds $w'$ from $Aw$ must have $\norm{2}{w'-w}^2 > (1
  - \eps) \norm{2}{w}^2$ with probability at least $1-O(\delta)$.
\end{lemma}
\begin{proof}
  Note that $Aw$ merely gives the projection of $w$ onto $m$
  dimensions, giving no information about the other $n-m$ dimensions.
  Since $w$ and the $\ell_2$ norm are rotation invariant, we may
  assume WLOG that $A$ gives the projection of $w$ onto the first $m$
  dimensions, namely $T = [m]$.  By the norm concentration of
  Gaussians, with probability $1-\delta$ we have $\norm{2}{w}^2 <
  (1+\eps)n\sigma^2$, and by Markov with probability $1-\delta$ we
  have $\norm{2}{w_T}^2 < \eps n \sigma^2$.

  For any fixed value $d$, since $w$ is uniform Gaussian and
  $w'_{\overline{T}}$ is independent of $w_{\overline{T}}$,
  \begin{align*}
    \Pr[\norm{2}{w'-w}^2 < d] &\leq \Pr[\norm{2}{(w'-w)_{\overline{T}}}^2 < d]
%    \ifconf\notag\\&\fi %XXX remove in non-conf
\leq \Pr[\norm{2}{w_{\overline{T}}}^2 < d].
  \end{align*}
  Therefore
  \begin{align*}
%    \ifconf\notag&\fi %XXX remove in non-conf
    \Pr[\norm{2}{w'-w}^2 < (1-3\eps)\norm{2}{w}^2]
%    \ifconf\notag\\\fi %XXX remove in non-conf
    \leq&
    \Pr[\norm{2}{w'-w}^2 < (1-2\eps)n\sigma^2] \\\leq&
    \Pr[\norm{2}{w_{\overline{T}}}^2 < (1-2\eps)n\sigma^2] \\\leq&
    \Pr[\norm{2}{w_{\overline{T}}}^2 < (1-\eps)(n-m)\sigma^2] \leq
    \delta
  \end{align*}
  as desired.  Rescaling $\eps$ gives the result.
\end{proof}

\begin{lemma}\label{lemma:inflower}
  Suppose $n = \Omega(1/\eps^2 + (k/\eps)\log (k/\eps))$ and $m = O(\eps n)$.  Then $I(S; S')
  = \Omega(k \log (n/k))$ for some $\alpha = \Omega(1/\eps)$.
\end{lemma}
\begin{proof}
  Consider the $x'$ recovered from $A(x + w)$, and let $T = S \cup
  S'$.  Suppose that $\norm{\infty}{w}^2 \leq O(\frac{\alpha k}{n}\log
  n)$ and $\norm{2}{w}^2 / (\alpha k) \in [1
  \pm \eps]$, as happens with probability at least (say) $3/4$.  Then
  we claim that if recovery is successful, one of the following must
  be true:
  \begin{align}
    \norm{2}{x'_T-x}^2 &\leq 9\eps \norm{2}{w}^2\label{eq:closetox}\\
    \norm{2}{x'_{\overline{T}}-w}^2 &\leq (1-2\eps) \norm{2}{w}^2\label{eq:closetow}
  \end{align}
  To show this, suppose $\norm{2}{x'_T-x}^2 > 9\eps \norm{2}{w}^2 \geq
  9\norm{2}{w_T}^2$ (the last by $\abs{T} = 2k = O(\eps n / \log n)$).  Then
  \begin{align*}
    \norm{2}{(x' - (x + w))_T}^2 &> (\norm{2}{x'-x}-
    \norm{2}{w_T})^2\\
    &\geq (2\norm{2}{x'-x}/3)^2 \geq 4\eps\norm{2}{w}^2.
  \end{align*}
  Because recovery is successful,
  \[
  \norm{2}{x' - (x + w)}^2 \leq  (1+\eps)\norm{2}{w}^2.
  \]
  Therefore
  \begin{align*}
    \norm{2}{x'_{\overline{T}}-w_{\overline{T}}}^2 + \norm{2}{x'_T -
      (x+w)_T}^2 &= \norm{2}{x' - (x + w)}^2\\
    \norm{2}{x'_{\overline{T}}-w_{\overline{T}}}^2 + 4\eps\norm{2}{w}^2 &< (1+\eps)\norm{2}{w}^2\\
    \norm{2}{x'_{\overline{T}}-w}^2 - \norm{2}{w_T}^2 &< (1-3\eps)\norm{2}{w}^2 
%    \ifconf\notag\\&\fi %XXX remove in non-conf
    \leq (1-2\eps)\norm{2}{w}^2
  \end{align*}
  as desired.  Thus with $3/4$ probability, at least one
  of~\eqref{eq:closetox} and \eqref{eq:closetow} is true.

  Suppose Equation~\eqref{eq:closetow} holds with at least $1/4$
  probability.  There must be some $x$ and $S$ such that the same
  equation holds with $1/4$ probability.  For this $S$, given $x'$ we
  can find $T$ and thus $x'_{\overline{T}}$.  Hence for a uniform
  Gaussian $w_{\overline{T}}$, given $Aw_{\overline{T}}$ we can
  compute $A(x + w_{\overline{T}})$ and recover $x'_{\overline{T}}$
  with $\norm{2}{x'_{\overline{T}} - w_{\overline{T}}}^2 \leq
  (1-\eps)\norm{2}{w_{\overline{T}}}^2$.  By
  Lemma~\ref{lemma:gaussianhard} this is impossible, since $n-\abs{T}
  = \Omega(\frac{1}{\eps^2})$ and $m = \Omega(\eps n)$ by assumption.

  Therefore Equation~\eqref{eq:closetox} holds with at least $1/2$
  probability, namely $\norm{2}{x'_T-x}^2 \leq 9\eps\norm{2}{w}^2 \leq
  9\eps(1-\eps) \alpha k < k/2$ for appropriate $\alpha$.  But if the
  nearest $\hat{x} \in X$ to $x$ is not equal to $x$,
  \begin{align*}
%    \ifconf&\fi %XXX remove in non-conf
    \norm{2}{x' - \hat{x}}^2 
    \ifconf\\\fi
    =& \norm{2}{x'_{\overline{T}}}^2 +
    \norm{2}{x'_{\overline{T}} - \hat{x}}^2 \geq
    \norm{2}{x'_{\overline{T}}}^2 + (\norm{2}{x - \hat{x}} -
    \norm{2}{x'_{\overline{T}} - x})^2 \\>&
    \norm{2}{x'_{\overline{T}}}^2 + (k - k/2)^2 >
    \norm{2}{x'_{\overline{T}}}^2 + \norm{2}{x'_{\overline{T}} - x}^2
    = \norm{2}{x'-x}^2,
  \end{align*}
  a contradiction.  Hence $S' = S$.  But Fano's inequality states
  $H(S | S') \leq 1 + \Pr[S' \neq S]\log \abs{\mathcal{F}}$
and hence
  \[
  I(S; S') = H(S) - H(S | S') \geq -1 + \frac{1}{4} \log \abs{\mathcal{F}} = \Omega(k \log (n/k))
  \]
  as desired.
\end{proof}
\begin{theorem}
  Any $(1+\eps)$-approximate $\ell_2/\ell_2$ recovery scheme with
  $\eps > \sqrt{\frac{k\log n}{n}}$ and failure probability $\delta <
  1/2$ requires $m = \Omega(\frac{1}{\eps}k\log (n/k))$.
\end{theorem}
\begin{proof}
  Combine Lemmas~\ref{lemma:inflower} and~\ref{lemma:infupper} with
  $\alpha = 1/\eps$ to get $m = \Omega(\frac{k \log (n/k)}{\log (1 +
    \eps)}) = \Omega(\frac{1}{\eps}k\log (n/k))$, $m = \Omega(\eps
  n)$, or $n = O(\frac{1}{\eps}k \log (k/\eps))$.  For $\eps$ as in the
  theorem statement, the first bound is controlling.
\end{proof}

\section{Bit complexity to measurement complexity}

The remaining lower bounds proceed by reductions from communication
complexity.  The following lemma (implicit in~\cite{DIPW}) shows that
lower bounding the number of bits for approximate recovery is
sufficient to lower bound the number of measurements.  Let $B_p^n(R)
\subset \R^n$ denote the $\ell_p$ ball of radius $R$.

\begin{definition}
  Let $X \subset \R^n$ be a distribution with $x_i \in \{-n^d, \dotsc,
  n^d\}$ for all $i \in [n]$ and $x \in X$.  We define a
  $1+\eps$-approximate $\ell_p/\ell_p$ sparse recovery \emph{bit
    scheme} on $X$ with $b$ bits, precision $n^{-c}$, and failure
  probability $\delta$ to be a deterministic pair of functions $f
  \colon X \to \{0,1\}^b$ and $g\colon \{0, 1\}^b \to \R^n$ where $f$
  is linear so that $f(a+b)$ can be computed from $f(a)$ and $f(b)$.
  We require that, for $u \in B_p^n(n^{-c})$ uniformly and $x$ drawn
  from $X$, $g(f(x))$ is a valid result of $1+\eps$-approximate
  recovery on $x+u$ with probability $1-\delta$.
\end{definition}

\begin{lemma}\label{lemma:bitstomeasurements}
  A lower bound of $\Omega(b)$ bits for such a sparse recovery bit
  scheme with $p \leq 2$ implies a lower bound of
  $\Omega(b/((1+c+d)\log n))$ bits for regular $(1+\eps)$-approximate
  sparse recovery with failure probability $\delta - 1/n$.
\end{lemma}

\begin{proof}
  Suppose we have a standard $(1+\eps)$-approximate sparse recovery
  algorithm $\mathcal{A}$ with failure probability $\delta$ using $m$
  measurements $Ax$.  We will use this to construct a (randomized)
  sparse recovery bit scheme using $O(m(1+c+d)\log n)$ bits and
  failure probability $\delta + 1/n$.  Then by averaging some
  deterministic sparse recovery bit scheme performs better than
  average over the input distribution.

  We may assume that $A \in \R^{m \times n}$ has orthonormal rows
  (otherwise, if $A = U\Sigma V^T$ is its singular value
  decomposition, $\Sigma^{+}U^TA$ has this property and can be
  inverted before applying the algorithm).  When applied to the
  distribution $X + u$ for $u$ uniform over $B_p^n(n^{-c})$, we may
  assume that $\mathcal{A}$ and $A$ are deterministic and fail with
  probability $\delta$ over their input.

  Let $A'$ be $A$ rounded to $t\log n$ bits per entry for some
  parameter $t$.  Let $x$ be chosen from $X$.  By Lemma~5.1
  of~\cite{DIPW}, for any $x$ we have $A'x = A(x-s)$ for some $s$ with
  $\norm{1}{s} \leq n^22^{-t\log n}\norm{1}{x}$, so $\norm{p}{s} \leq
  n^{2.5-t} \norm{p}{x} \leq n^{3.5+d-t}$. Let $u \in
  B_p^n(n^{5.5+d-t})$ uniformly at random.  With probability at least
  $1-1/n$, $u \in B_p^n((1-1/n^2)n^{5.5+d-t})$ because the balls are
  similar so the ratio of volumes is $(1-1/n^2)^n > 1-1/n$.  In this
  case $u + s \in B_p^n(n^{5.5+d-t})$; hence the random variable $u$
  and $u + s$ overlap in at least a $1 - 1/n$ fraction of their
  volumes, so $x + s + u$ and $x + u$ have statistical distance at
  most $1/n$.  Therefore $\mathcal{A}(A(x+u)) = \mathcal{A}(A'x+Au)$
  with probability at least $1-1/n$.

  Now, $A'x$ uses only $(t+d+1)\log n$ bits per entry, so we can set
  $f(x) = A'x$ for $b = m(t+d+1)\log n$.  Then we set $g(y) =
  \mathcal{A}(y + Au)$ for uniformly random $u \in
  B_p^n(n^{5.5+d-t})$.  Setting $t = 5.5 + d + c$, this gives a sparse
  recovery bit scheme using $b = m(6.5 + 2d + c)\log n$.
\end{proof}

\section{Non-sparse output Lower Bound for $p=1$}

\iffalse
Sketch: same as previous section, with $\alpha = 1/\eps$ and $n =
O(k/\eps^{3/2})$, so $\E[\abs{w_i}] = O(\sqrt{\eps})$.  The $\ell_1/\ell_1$
guarantee gives
\begin{align*}
  \norm{1}{x'-(x+w)} &\leq (1+\eps)\norm{1}{w}\\
  \norm{1}{(x'-x)_T} + \norm{1}{x'_{\overline{T}}-w} &\leq (1+\eps)\norm{1}{w} + 2\norm{1}{w_T} 
%\ifconf\\&\fi %XXX remove in non-conf
\leq (1+3\eps)\norm{1}{w}
\end{align*}
with probability $1 - O(\eps^{1/2})$, since then $\norm{1}{w_T} < \eps
\norm{1}{w}$.

Thus either $\norm{1}{(x'-x)_T} < O(\eps) \norm{1}{w} < k/2$, or
$\norm{1}{x'_{\overline{T}} - w} \leq (1-\eps)\norm{1}{w}$.

\xxx{Not done}
\subsection{Old lower bound for non-sparse output}
\fi
%This section is divided into three parts.  
First, we show that
recovering the locations of an $\eps$ fraction of $d$ ones in a vector
of size $n > d/\eps$ requires $\widetilde\Omega(\eps d)$ bits.  Then, we show
high bit complexity of a distributional product version of the
Gap-$\ell_\infty$ problem.  Finally, we create a distribution for
which successful sparse recovery must solve one of the previous
problems, giving a lower bound in bit complexity.
Lemma~\ref{lemma:bitstomeasurements} converts the bit complexity to
measurement complexity.

\subsection{$\ell_1$ Lower bound for recovering noise bits} %$\eps$ fraction of $\text{poly}(1/\eps)$ bits}

\begin{definition}
  We say a set $C \subset [q]^d$ is a $(d, q, \eps)$ code if any two
  distinct $c, c' \in C$ agree in at most $\eps d$ positions.  We say
  a set $X \subset \{0, 1\}^{dq}$ represents $C$ if $X$ is $C$
  concatenated with the trivial code $[q] \to \{0, 1\}^q$ given by $i
  \to e_i$.
\end{definition}

\begin{claim}\label{claim:codesize}
  For $\eps \geq 2/q$, there exist $(d, q, \eps)$ codes $C$ of size
  $q^{\Omega(\eps d)}$ by the Gilbert-Varshamov bound (details
  in~\cite{DIPW}).
\end{claim}

\begin{lemma}\label{lemma:l1polys}
  Let $X \subset \{0, 1\}^{dq}$ represent a $(d, q, \eps)$ code.
  Suppose $y \in \R^{dq}$ satisfies $\norm{1}{y-x} \leq (1 -
  \eps)\norm{1}{x}$.  Then we can recover $x$ uniquely from $y$.
\end{lemma}
\begin{proof}
  We assume $y_i \in [0,1]$ for all $i$; thresholding otherwise
  decreases $\norm{1}{y-x}$.  We will show that there exists no other
  $x' \in X$ with $\norm{1}{y-x} \leq (1-\eps)\norm{1}{x}$; thus
  choosing the nearest element of $X$ is a unique decoder.  Suppose
  otherwise, and let $S = \supp(x), T = \supp(x')$.  Then
  \begin{align*}
    (1 - \eps)\norm{1}{x} &\geq \norm{1}{x-y} \\
    &= \norm{1}{x} - \norm{1}{y_S} + \norm{1}{y_{\overline{S}}}\\
    \norm{1}{y_S} &\geq \norm{1}{y_{\overline{S}}} + \eps d
  \end{align*}
  Since the same is true relative to $x'$ and $T$, we have
  \begin{align*}
    \norm{1}{y_S} + \norm{1}{y_T} &\geq \norm{1}{y_{\overline{S}}} + \norm{1}{y_{\overline{T}}} + 2\eps d\\
    2\norm{1}{y_{S \cap T}}&\geq 2\norm{1}{y_{\overline{S \cup T}}} + 2\eps d\\
    \norm{1}{y_{S \cap T}}&\geq \eps d\\
    \abs{S \cap T}&\geq \eps d
  \end{align*}
  This violates the distance of the code represented by $X$.
\end{proof}

\begin{lemma}\label{lemma:l1entropy}
  Let $R = [s, cs]$ for some constant $c$ and parameter $s$.  Let $X$
  be a permutation independent distribution over $\{0, 1\}^{n}$ with
  $\norm{1}{x} \in R$ with probability $p$.  If $y$ satisfies
  $\norm{1}{x-y} \leq (1-\eps)\norm{1}{x}$ with probability $p'$ with
  $p' - (1-p) = \Omega(1)$, then $I(x; y) = \Omega(\eps s\log (n/s))$.
\end{lemma}
\begin{proof}
  For each integer $i \in R$, let $X_i \subset \{0, 1\}^n$ represent
  an $(i, n/i, \eps)$ code.  Let $p_i = \Pr_{x \in X}[\norm{1}{x} =
  i]$.  Let $S_n$ be the set of permutations of $[n]$.  Then the
  distribution $X'$ given by (a) choosing $i \in R$ proportional to
  $p_i$, (b) choosing $\sigma \in S_n$ uniformly, (c) choosing $x_i
  \in X_i$ uniformly, and (d) outputting $x'=\sigma(x_i)$ is equal to
  the distribution $(x\in X \mid \norm{1}{x} \in R)$.

  Now, because $p' \geq \Pr[\norm{1}{x} \notin R] + \Omega(1)$, $x'$
  chosen from $X'$ satisfies $\norm{1}{x'-y} \leq
  (1-\eps)\norm{1}{x'}$ with $\delta \geq p'- (1-p)$ probability.
  Therefore, with at least $\delta/2$ probability, $i$ and $\sigma$
  are such that $\norm{1}{\sigma(x_i)-y} \leq
  (1-\eps)\norm{1}{\sigma(x_i)}$ with $\delta/2$ probability over
  uniform $x_i \in X_i$.  But given $y$ with $\norm{1}{y-\sigma(x_i)}$
  small, we can compute $y' = \sigma^{-1}(y)$ with $\norm{1}{y' -
    x_i}$ equally small.  Then by Lemma~\ref{lemma:l1polys} we can
  recover $x_i$ from $y$ with probability $\delta/2$ over $x_i \in
  X_i$.  Thus for this $i$ and $\sigma$, $I(x; y \mid i, \sigma) \geq
  \Omega(\log \abs{X_i}) = \Omega(\delta \eps s \log (n/s))$ by Fano's
  inequality.  But then $I(x; y) = \E_{i, \sigma}[I(x; y \mid i,
  \sigma)] = \Omega(\delta^2\eps s \log (n/s)) = \Omega(\eps s \log
  (n/s))$.
\end{proof}

\subsection{Distributional Indexed Gap $\ell_\infty$}
%
%We first define and analyze the complexity of what we call the Indexed
%$\ell_{\infty}$ problem.  First, 
Consider the following communication
game, which we refer to as $\gaplinf^B$, studied in \cite{BJKS04}. The
legal instances are pairs $(x,y)$ of $m$-dimensional vectors, with
$x_i, y_i \in \{0, 1, 2, \ldots, B\}$ for all $i$ such that
\begin{itemize}
\item NO instance: for all $i$, $y_i-x_i \in \{0,1\}$, or
\item YES instance: there is a \emph{unique} $i$ for which
  $y_i - x_i = B$, and for all $j \neq i$, $y_i-x_i \in \{0,1\}$.
%\footnote{One should note that the usual definition of $\gaplinf^B$ is a little
%less restrictive on the values of $y_i - x_i$, but the lower bound
%in \cite{BJKS04} applies even with such restrictions since the distribution
%they consider to prove the lower bound has such restrictions.
%See e.g., p.176 of \cite{BarYossefThesis} 
%in the proof of 
%Theorem 6.32.}
%, where the hard distribution $\mu$ is such that $y_i - x_i \in \{0,1\}$
%for all coordinates $i$. Indeed,  
%in the author's notation, a random $d$ is chosen and then 
%$(x_i, y_i) \in \{(d-1,d), (d,d), (d, d+1)\}$. Hence, we can define the problem so that
%$y_i-x_i \in \{0,1\}$ whenever $y_i - x_i \neq B$.}
\end{itemize}
%The randomized communication complexity of $\gaplinf^{B}$ is well-understood.
%\begin{theorem}(\cite{BJKS04})\label{thm:bjks} For $\delta < 1/3$,
%  the randomized $\delta$-error communication cost
%  $R_{\delta}(\gaplinf^{B}) = \Omega(m/B^2)$.
%\end{theorem}
%For our applications, we need several finer notions of
%communication complexity. 
The {\it distributional}
communication complexity $D_{\sigma, \delta}(f)$ of a function $f$ is the minimum
over all deterministic protocols computing $f$ with error probability at most $\delta$, where
the probability is over inputs drawn from $\sigma$. 
%The {\it amortized} communication complexity $D_{\sigma, \delta}^r(f)$ \cite{BR11}of a function $f$
%is the minimum over all deterministic protocols that receive $r$ independent pairs of
%inputs drawn according to $\sigma$, and which get the answer correctly with probability
%at least $1-\delta$ on each pair. 

Consider the distribution $\sigma$ which chooses a random $i \in [m]$. Then for each $j \neq i$, 
it chooses a random $d \in \{0, \ldots, B\}$ and $(x_i, y_i)$ is uniform in $\{(d, d), (d, d+1)\}$. 
For coordinate $i$, $(x_i, y_i)$ is uniform in $\{(0,0), (0,B)\}$. Using similar arguments to those
in \cite{BJKS04}, Jayram \cite{J02} showed $D_{\sigma, \delta}(\gaplinf^B) = \Omega(m/B^2)$ (this is
reference [70] on p.182 of \cite{BarYossefThesis}) for $\delta$ less than a small constant. 

We define the one-way distributional communication complexity $D^{1-way}_{\sigma, \delta}(f)$ of a
function $f$ to be the smallest distributional complexity of a protocol for $f$
in which only a single message is sent from Alice to Bob.
%We define the one-way amortized communication complexity
%$D^{1-way, r}_{\sigma, \delta}(f)$ to be amortized communication complexity of the best deterministic
%protocol for $f$ in which only a single message is sent from Alice to Bob.
%
%\begin{theorem}(special case of Corollary 2.5 of
%  \cite{BR11})\label{thm:bo} 
%  Suppose $D_{\sigma, \delta}(f) \geq C$ for a sufficiently large constant $C > 0$. 
%  Then $D^{1-way, r}_{\sigma, \delta/2}(f) = \Omega(n(C-\log(2/\delta)))$. 
%\end{theorem}
%
%We also need the following problem.
\begin{definition}[Indexed $\indlinf^{r,B}$ Problem] There are
  $r$ pairs of inputs $(x^1,y^1),(x^2,y^2),\ldots,(x^r,y^r)$ such that
  every pair $(x^i,y^i)$ is a legal instance of the $\gaplinf^B$
  problem. Alice is given $x^1, \ldots, x^{r}$.  Bob is given
  an index $I \in [r]$ and $y^1, \ldots, y^r$.  The goal is to
  decide whether $(x^I, y^I)$ is a NO or a YES instance of
  $\gaplinf^B$.
\end{definition}
Let $\eta$ be the distribution $\sigma^r \times U_r$, where $U_r$ is the
uniform distribution on $[r]$. We bound
$D^{1-way}_{\eta, \delta}(\indlinf)^{r,B}$ as follows. 
For a function $f$, let $f^r$ denote the problem of computing $r$ instances of $f$. 
For a distribution $\zeta$ on instances of $f$,  
let $D_{\zeta^r, \delta}^{1-way, *} (f^r)$ denote the minimum communication cost of a deterministic
protocol computing a function $f$ with error probability at most $\delta$ {\it in each} 
of the $r$ copies of $f$, where the inputs come from $\zeta^r$. 

\begin{theorem}\label{thm:bo}(special case of Corollary 2.5 of \cite{br11})
Assume $D_{\sigma, \delta}(f)$ is larger than a large enough constant. Then 
$D^{1-way, *}_{\sigma^r, \delta/2}(f^r) = \Omega(r D_{\sigma, \delta}(f))$.
\end{theorem}

\begin{theorem}\label{thm:main}
For $\delta$ less than a sufficiently small constant, 
$D^{1-way}_{\eta, \delta}(\indlinf^{r,B}) = \Omega(\delta^2 r m/(B^2 \log r))$.
\end{theorem}
\begin{proof} 
Consider a deterministic $1$-way protocol $\Pi$ for $\indlinf^{r,B}$ with error probability
$\delta$ on inputs drawn from $\eta$. Then for at least $r/2$ values $i \in [r]$,
$\Pr[\Pi(x^1, \ldots, x^r, y^1, \ldots, y^r, I) = \gaplinf^B(x^I, y^I) \mid I = i] \geq 1-2\delta.$
Fix a set $S = \{i_1, \ldots, i_{r/2}\}$ of indices with this property. We build a deterministic
$1$-way protocol $\Pi'$ for $f^{r/2}$ with input distribution $\sigma^{r/2}$ and error probability at
most $6 \delta$ in each of the $r/2$ copies of $f$. 

For each $\ell \in [r] \setminus S$, independently choose $(x^{\ell}, y^{\ell}) \sim \sigma$. 
For each $j \in [r/2]$, let $Z_j^1$ be the probability that 
$\Pi(x^1, \ldots, x^r, y^1, \ldots, y^r, I) =  \gaplinf^B(x^{i_j}, y^{i_j})$ given $I = i_j$ 
and the choice of $(x^{\ell}, y^{\ell})$ for all $\ell \in [r] \setminus S$. 

If we repeat this experiment 
independently $s = O(\delta^{-2}\log r)$
times, obtaining independent $Z_j^1, \ldots, Z_j^s$ and let $Z_j = \sum_t Z_j^t$, then
$\Pr[Z_j \geq s - s \cdot 3\delta] \geq 1- \frac{1}{r}.$
So there exists a set of $s = O(\delta^{-1} \log r)$ repetitions for which for each $j \in [r/2]$,
$Z_j \geq s- s \cdot 3 \delta$. 
We hardwire these into $\Pi'$ to make the protocol deterministic. 

Given inputs $((X^1, \ldots, X^{r/2}), (Y^1, \ldots, Y^{r/2})) \sim \sigma^{r/2}$ to $\Pi'$, Alice and Bob 
run $s$ executions of $\Pi$, each with $x^{i_j} = X^j$ and $y^{i_j} = Y^j$ for all $j \in [r/2]$, 
filling in the remaining values using the hardwired inputs. Bob runs the algorithm specified by $\Pi$
for each $i_j \in S$ and each execution. His output for $(X^j, Y^j)$ is the majority
of the outputs of the $s$ executions with index $i_j$. 

Fix an index $i_j$. Let $W$ be the number of repetitions for which 
$\gaplinf^B(X^j, Y^j)$ does not equal the output of $\Pi$ on input $i_j$, for a random $(X^j, Y^j) \sim \sigma$.
Then, ${\bf E}[W] \leq 3\delta$. 
By a Markov bound, $\Pr[W \geq s/2] \leq 6\delta$, and so the coordinate is correct
with probability at least $1-6\delta$. 

The communication of $\Pi'$ is a 
factor $s = \Theta(\delta^{-2} \log r)$ more than that of $\Pi$. The theorem now follows by Theorem 
\ref{thm:bo}, using that $D_{\sigma, 12\delta}(\gaplinf^B) = \Omega(m/B^2)$. 
\end{proof}

\subsection{Lower bound for sparse recovery}
Fix the parameters $B = \Theta(1/\eps^{1/2}), r = k$, $m =
1/\eps^{3/2}$, and $n=k/\eps^3$.  Given an instance $(x^1, y^1),
\ldots, (x^r, y^r), I$ of $\indlinf^{r,B}$, we define the input signal
$z$ to a sparse recovery problem.  We allocate a set $S^i$ of $m$
disjoint coordinates in a universe of size $n$ for each pair $(x^i,
y^i)$, and on these coordinates place the vector $y^i-x^i$. The
locations are important for arguing the sparse recovery algorithm
cannot learn much information about the noise, and will be placed
uniformly at random.
%
%turns out to be crucial to solve
%our communication problems from the solution of the sparse recovery
%algorithm. The placement is correlated across the different $S^i$. For
%details on the way we place the items, see Lemma \ref{lem:entropy}
%below. For most of the discussion, the placement will not matter.

Let $\rho$ denote the induced distribution on $z$. Fix a
$(1+\eps)$-approximate $k$-sparse recovery bit scheme $Alg$ that takes
$b$ bits as input and succeeds with probability at least $1-\delta/2$
over $z \sim \rho$ for some small constant $\delta$.  Let $S$ be the
set of top $k$ coordinates in $z$. $Alg$ has the guarantee that if it
succeeds for $z \sim \rho$, then there exists a small $u$ with
$\norm{1}{u} < n^{-2}$ so that $v = Alg(z)$ satisfies
\begin{align}
\norm{1}{v-z-u} &\leq (1+\eps)\norm{1}{(z+u)_{[n]\setminus S}}\notag\\
\norm{1}{v-z} &\leq (1+\eps)\norm{1}{z_{[n]\setminus S}} + (2+\eps)/n^2 \notag\\
&\leq (1+2\eps)\norm{1}{z_{[n]\setminus S}}\notag
\end{align}
and thus
\begin{align}\label{eqn:bound}
\norm{1}{(v-z)_S} + \norm{1}{(v-z)_{[n] \setminus S}}
\leq (1+2\eps)\|z_{[n] \setminus S}\|_1.
\end{align}

\begin{lemma}\label{lem:heavyHitters}
  For $B= \Theta(1/\eps^{1/2})$ sufficiently large, suppose that
  $\Pr_{z \sim \rho}[\|(v-z)_S\|_1 \leq 10 \eps \cdot \|z_{[n]
    \setminus S}\|_1] \geq 1-\delta$.  Then $Alg$ requires $b =
  \Omega(k/(\eps^{1/2} \log k))$.
\end{lemma}
\begin{proof}
  %By Lemma \ref{lemma:bitstomeasurements}, we can assume each
  %measurement of $Alg$ can uses $O(\log n) = O(\log
  %1/\eps)$ bits, and so it suffices to show the bit complexity
  %of the state representing $Alg$ is $\Omega(k/\eps^{1/2})$. To
  %do so, 
  We show how to use $Alg$ to solve instances of $\indlinf^{r,B}$ with
  probability at least $1- C$ for some small $C$, where the
  probability is over input instances to $\indlinf^{r,B}$ distributed
  according to $\eta$, inducing the distribution $\rho$. The lower
  bound will follow by Theorem \ref{thm:main}.  Since $Alg$ is a
  deterministic sparse recovery bit scheme, it receives a sketch
  $f(z)$ of the input signal $z$ and runs an arbitrary recovery
  algorithm $g$ on $f(z)$ to determine its output $v = Alg(z)$.

  Given $x^1, \ldots, x^r$, for each $i = 1, 2, \ldots, r$, Alice places
  $-x^i$ on the appropriate coordinates in the block $S^i$ used in
  defining $z$, obtaining a vector $z_{Alice}$, and transmits
  $f(z_{Alice})$ to Bob.  Bob uses his inputs $y^1, \ldots, y^r$ to
  place $y^i$ on the appropriate coordinate in $S^i$. He thus creates a
  vector $z_{Bob}$ for which $z_{Alice} + z_{Bob} = z$. Given
  $f(z_{Alice})$, Bob computes $f(z)$ from $f(z_{Alice})$ and
  $f(z_{Bob})$, then $v = Alg(z)$. We assume all coordinates of $v$
  are rounded to the real interval $[0, B]$, as this can only decrease
  the error.

  %By definition of $z$, we can think of there being $k$ blocks
  %$B_1, \ldots, B^k$ of coordinates, one for each pair $(x^i,
  %y^i)$. 
  We say that $S^i$ is {\it bad} if either
  \begin{itemize}
%  \item there are at least two coordinates $j,j'$ in $S^i$ for which
%    $|v_j| \geq B/2$ and $|v_{j'}| \geq B/2$
  \item there is no coordinate $j$ in $S^i$ for which $|v_j| \geq
    \frac{B}{2}$ yet $(x^i, y^i)$ is a YES instance of
    $\gaplinf^{r,B}$, or 
  \item there is a coordinate $j$ in $S^i$ for which $|v_j| \geq
    \frac{B}{2}$ yet either $(x^i, y^i)$ is a NO instance of
    $\gaplinf^{r, B}$ or $j$ is not the unique $j^*$ for which $y^i_{j^*} - x^i_{j^*} = B$
  \end{itemize}
  The $\ell_1$-error incurred by a bad block is at least
  $B/2-1$. Hence, if there are $t$ bad blocks, the total error is at
  least $t(B/2-1)$, which must be smaller than $10 \eps \cdot \|z_{[n]
    \setminus S}\|_1$ with probability $1-\delta$.  Suppose this
  happens.

  We bound $t$.  All coordinates in $z_{[n] \setminus S}$ have value
  in the set $\{0, 1\}$. Hence, $\|z_{[n] \setminus S}\|_1 < rm$.  So
  $t \leq 20 \eps rm /(B-2)$.  For $B \geq 6$, $t \leq 30 \eps rm/B$.
  Plugging in $r$, $m$ and
  $B$, %$t \leq \frac{Ck}{\eps^{3/2}} \cdot 10 \eps \cdot C \cdot (3
  %\eps^{1/2}),$ where $C > 0$ is a constant that can be made
  %arbitrarily small by increasing $B = \Theta(1/\eps^{1/2})$.  Hence,
  $t \leq C k$, where $C > 0$ is a constant that can be made
  arbitrarily small by increasing $B = \Theta(1/\eps^{1/2})$.

  If a block $S^i$ is not bad, then it can be used
  to solve $\gaplinf^{r, B}$ on $(x^i, y^i)$ with probability
  $1$. Bob declares that $(x^i, y^i)$ is a YES instance if and
  only if there is a coordinate $j$ in $S^i$ for which $|v_j| \geq
  B/2$.

  Since Bob's index $I$ is uniform on the $m$ coordinates in
  $\indlinf^{r,B}$, with probability at least $1-C$ the players solve
  $\indlinf^{r, B}$ given that the $\ell_1$ error is small. Therefore
  they solve $\indlinf^{r, B}$ with probability $1-\delta-C$
  overall. By Theorem \ref{thm:main}, for $C$ and $\delta$
  sufficiently small $Alg$ requires $\Omega(mr/(B^2 \log r)) =
  \Omega(k/(\eps^{1/2} \log k))$ bits.
\end{proof}

\begin{lemma}\label{lem:entropy}
  Suppose $\Pr_{z \sim \rho}[\|(v-z)_{[n] \setminus S}\|_1] \leq
  (1-8\eps)\cdot \|z_{[n] \setminus S}\|_1] \geq \delta/2$.  Then
  $Alg$ requires $b = \Omega(\frac{1}{\sqrt{\eps}}k \log
  (1/\eps))$.
\end{lemma}
\begin{proof}
  The distribution $\rho$ consists of $B(mr, 1/2)$ ones placed
  uniformly throughout the $n$ coordinates, where $B(mr, 1/2)$ denotes
  the binomial distribution with $mr$ events of $1/2$ probability each.
  Therefore with probability at least $1-\delta/4$, the number of ones
  lies in $[\delta mr / 8, (1 - \delta/8)mr]$.  Thus by
  Lemma~\ref{lemma:l1entropy}, $I(v; z) \geq \Omega(\eps mr \log
  (n/(mr)))$.  Since the mutual information only passes through a $b$-bit
  string, $b = \Omega(\eps mr \log (n/(mr)))$ as well.
\end{proof}

\begin{theorem}
  Any $(1+\eps)$-approximate $\ell_1/\ell_1$ recovery scheme with
  sufficiently small constant failure probability $\delta$ must make
  $\Omega(\frac{1}{\sqrt{\eps}}k/\log^2 (k/\eps))$ measurements.
\end{theorem}
\begin{proof}
  We will lower bound any $\ell_1/\ell_1$ sparse recovery bit scheme
  $Alg$. If $Alg$ succeeds, then in order to satisfy inequality
  (\ref{eqn:bound}), we must either have $\|(v-z)_S\|_1 \leq 10 \eps
  \cdot \|z_{[n] \setminus S}\|_1$ or we must have $\|(v-z)_{[n]
    \setminus S}\|_1 \leq (1-8\eps)\cdot \|z_{[n] \setminus
    S}\|_1$. Since $Alg$ succeeds with probability at least
  $1-\delta$, it must either satisfy the hypothesis of Lemma
  \ref{lem:heavyHitters} or the hypothesis of Lemma
  \ref{lem:entropy}. But by these two lemmas, it follows that $b =
  \Omega(\frac{1}{\sqrt{\eps}}k/\log k)$.  Therefore by
  Lemma~\ref{lemma:bitstomeasurements}, any $(1+\eps)$-approximate
  $\ell_1/\ell_1$ sparse recovery algorithm requires
  $\Omega(\frac{1}{\sqrt{\eps}}k/\log^2 (k/\eps))$ measurements.
\end{proof}

\section{Lower bounds for $k$-sparse output}

\begin{theorem}\label{thm:sparsel1lower}
  Any $1+\eps$-approximate $\ell_1/\ell_1$ recovery scheme with
  $k$-sparse output and failure probability $\delta$ requires $m =
  \Omega(\frac{1}{\eps}(k\log \frac{1}{\eps} + \log
  \frac{1}{\delta}))$, for $32 \leq \frac{1}{\delta} \leq n\eps^2/k$.
\end{theorem}
\begin{theorem}\label{thm:sparsel2lower}
  Any $1+\eps$-approximate $\ell_2/\ell_2$ recovery scheme with
  $k$-sparse output and failure probability $\delta$ requires $m =
  \Omega(\frac{1}{\eps^2}(k + \log \frac{\eps^2}{\delta}))$, for $32 \leq
  \frac{1}{\delta} \leq n\eps^2/k$.
\end{theorem}

These two theorems correspond to four statements: one for large $k$
and one for small $\delta$ for both $\ell_1$ and $\ell_2$.

\ifconf All are fairly similar to the framework of~\cite{DIPW}: they
use a sparse recovery algorithm to robustly identify $x$ from $Ax$ for
$x$ in some set $X$.  This gives bit complexity $\log \abs{X}$, or
measurement complexity $\log \abs{X} / \log n$ by
Lemma~\ref{lemma:bitstomeasurements}.  They amplify the bit complexity
to $\log \abs{X} \log n$ by showing they can recover $x_1$ from $A(x_1
+ \frac{1}{10}x_2 + \dotsc + \frac{1}{n}x_{\Theta(\log n)})$ for $x_1,
\dotsc, x_{\Theta(\log n)} \in X$ and reducing from augmented
indexing.  This gives a $\log \abs{X}$ measurement lower bound.  Due
to space constraints, we defer full proof to the full paper. \else

All the lower bounds proceed by reductions from communication
complexity.  The following lemma (implicit in~\cite{DIPW}) shows that
lower bounding the number of bits for approximate recovery is
sufficient to lower bound the number of measurements.

\begin{lemma}\label{lemma:bitstomeasurements2}
  Let $p \in \{1, 2\}$ and $\alpha = \Omega(1) < 1$.  Suppose $X
  \subset \R^n$ has $\norm{p}{x} \leq D$ and $\norm{\infty}{x} \leq
  D'$ for all $x \in X$, and all coefficients of elements of $X$ are
  expressible in $O(\log n)$ bits.  Further suppose that we have a
  recovery algorithm that, for any $\nu$ with $\norm{p}{\nu} < \alpha
  D$ and $\norm{\infty}{\nu} < \alpha D'$, recovers $x \in X$ from
  $A(x + \nu)$ with constant probability.  Then $A$ must have
  $\Omega(\log \abs{X})$ measurements.
\end{lemma}

\begin{proof}
  \xxx{Use lemma~\ref{lemma:bitstomeasurements}}
  First, we may assume that $A \in \R^{m \times n}$ has orthonormal
  rows (otherwise, if $A = U\Sigma V^T$ is its singular value
  decomposition, $\Sigma^{+}U^TA$ has this property and can be
  inverted before applying the algorithm).  Let $A'$ be $A$ rounded to
  $c\log n$ bits per entry.  By Lemma~5.1 of~\cite{DIPW}, for any $v$
  we have $A'v = A(v-s)$ for some $s$ with $\norm{1}{s} \leq
  n^22^{-c\log n}\norm{1}{v}$, so $\norm{p}{s} \leq n^{2.5-c}
  \norm{p}{v}$.

  Suppose Alice has a bit string of length $r \log \abs{X}$ for $r =
  \Theta(\log n)$.  By splitting into $r$ blocks, this corresponds to
  $x_1, \dotsc, x_r \in X$.  Let $\beta$ be a power of $2$ between
  $\alpha/2$ and $\alpha/4$, and define
  \[
  z_j = \sum_{i=j}^r \beta^i x_i.
  \]
  Alice sends $A'z_1$ to Bob; this is $O(m\log n)$ bits.  Bob will
  solve the \emph{augmented indexing problem}\xxx{citation?}---given
  $A'z_1$, arbitrary $j \in [r]$, and $x_1, \dotsc, x_{j-1}$, he must
  find $x_j$ with constant probability.  This requires $A'z_1$ to have
  $\Omega(r \log \abs{X})$ bits, giving the result.

  Bob receives $A'z_1 = A(z_1 + s)$ for $\norm{1}{s} \leq
  n^{2.5-c}\norm{p}{z_1} \leq n^{2.5-c}D$.  Bob then chooses $u \in
  B_p^n(n^{4.5-c}D)$ uniformly at random.  With probability at least
  $1-1/n$, $u \in B_p^n((1-1/n^2)n^{4.5-c}D)$ by a volume argument.
  In this case $u + s \in B_p^n(n^{4.5-c}D)$; hence the random
  variables $u$ and $u + s$ overlap in at least a $1 - 1/n$ fraction
  of their volumes, so $z_j + s + u$ and $z_j + u$ have statistical
  distance at most $1/n$.  The distribution of $z_j + u$ is
  independent of $A$ (unlike $z_j + s$) so running the recovery
  algorithm on $A(z_j + s + u)$ succeeds with constant probability as
  well.

  We also have $\norm{p}{z_j} \leq \frac{\beta^j -
    \beta^{r+1}}{1-\beta}D < 2(\beta^j - \beta^{r+1})D$.  Since $r =
  O(\log n)$ and $\beta$ is a constant, there exists a $c = O(1)$ with
  \[
  \norm{p}{z_j + s + u} < (2\beta^j + n^{4.5-c} + n^{2.5-c}- 2\beta^r)D \leq \beta^{j-1}\alpha D
  \]
  for all $j$.

  Therefore, given $x_1, \dotsc, x_{j-1}$, Bob can compute
  \[
  \frac{1}{\beta^j}(A'z_1 + Au - A'\sum_{i < j} \beta^ix_i) = A(x_j + \frac{1}{\beta^j}(z_{j+1} + s + u)) = A(x_j + y)
  \]
  for some $y$ with $\norm{p}{y} \leq \alpha D$.  Hence Bob can use
  the recovery algorithm to recover $x_j$ with constant probability.
  Therefore Bob can solve augmented indexing, so the message $A'z_1$
  must have $\Omega(\log n \log \abs{X})$ bits, so $m = \Omega(\log
  \abs{X})$.
\end{proof}

We will now prove another lemma that is useful for all four theorem
statements.

Let $x \in \{0,1\}^n$ be $k$-sparse with $\supp(x) \subseteq S$ for
some known $S$.  Let $\nu \in \R^n$ be a noise vector that roughly
corresponds to having $O(k/\eps^p)$ ones for $p \in \{1,2\}$, all
located outside of $S$.  We consider under what circumstances we can
use a $(1+\eps)$-approximate $\ell_p/\ell_p$ recovery scheme to
recover $\supp(x)$ from $A(x + \nu)$ with (say) $90\%$ accuracy.

Lemma~\ref{lemma:recoverv} shows that this is possible for $p = 1$
when $\abs{S} \leq O(k / \eps)$ and for $p = 2$ when $\abs{S} \leq
2k$.  The algorithm in both instances is to choose a parameter $\mu$
and perform sparse recovery on $A(x + \nu + z)$, where $z_i = \mu$ for
$i \in S$ and $z_i = 0$ otherwise.  The support of the result will be
very close to $\supp(x)$.

\begin{lemma}\label{lemma:recoverv}
  Let $S \subset [n]$ have $\abs{S} \leq s$, and suppose $x \in
  \{0,1\}^n$ satisfies $\supp(x) \subseteq S$ and $\norm{1}{x_S} = k$.
  Let $p \in \{1,2\}$, and $\nu \in \R^n$ satisfy
  $\norm{\infty}{\nu_S} \leq \alpha$, $\norm{p}{\nu}^p \leq r$, and
  $\norm{\infty}{\nu} \leq D$ for some constants $\alpha \leq 1/4$ and
  $D = O(1)$.  Suppose $A \in \R^{m\times n}$ is part of a $(1 +
  \eps)$-approximate $k$-sparse $\ell_p/\ell_p$ recovery scheme with
  failure probability $\delta$.

  Then, given $A(x_S + \nu)$, Bob can with failure probability
  $\delta$ recover $\hat{x_S}$ that differs from $x_S$ in at most $k /
  c$ locations, as long as either
  \begin{align}
    p = 1, s = \Theta(\frac{k}{c \eps}), r = \Theta(\frac{k}{c \eps})
  \end{align}
  or
  \begin{align}
    p = 2, s = 2k, r = \Theta(\frac{k}{c^2 \eps^2})
  \end{align}
\end{lemma}
\begin{proof}
  For some parameter $\mu \geq D$, let $z_i = \mu$ for $i \in S$ and
  $z_i = 0$ elsewhere.  Consider $y = x_S + \nu + z$.  Let $U =
  \supp(x_S)$ have size $k$.  Let $V \subset [n]$ be the support of
  the result of running the recovery scheme on $Ay = A(x_S + \nu) +
  Az$.  Then we have that $x_S + z$ is $\mu + 1$ over $U$, $\mu$ over
  $S \setminus U$, and zero elsewhere.  Since $\norm{p}{u + v}^p \leq
  p(\norm{p}{u}^p + \norm{p}{v}^p)$ for any $u$ and $v$, we have
  \begin{align*}
    \norm{p}{y_{\overline{U}}}^p &\leq p(\norm{p}{(x_S + z)_{\overline{U}}}^p + \norm{p}{\nu}^p)\\
    &\leq p((s-k)\mu^p + r)\\
    &< p(r + s\mu^p).
  \end{align*}
  Since $\norm{\infty}{\nu_S} \leq \alpha$ and
  $\norm{\infty}{\nu_{\overline{S}}} < \mu$, we have
  \begin{align*}
    \norm{\infty}{y_U} &\geq \mu + 1 - \alpha\\
    \norm{\infty}{y_{\overline{U}}} &\leq \mu + \alpha  \end{align*}
  We then get
  \begin{align*}
    \norm{p}{y_{\overline{V}}}^p &= \norm{p}{y_{\overline{U}}}^p + \norm{p}{y_{U \setminus V}}^p  - \norm{p}{y_{V \setminus U}}^p\\
    &\geq \norm{p}{y_{\overline{U}}}^p + \abs{V \setminus U}((\mu + 1 - \alpha)^p - (\mu + \alpha)^p)\\
    &= \norm{p}{y_{\overline{U}}}^p + \abs{V \setminus U}(1 + (2p - 2)\mu)(1 - 2\alpha)
  \end{align*}
  where the last step can be checked for $p \in \{1, 2\}$.  So
  \begin{align*}
    \norm{p}{y_{\overline{V}}}^p &\geq \norm{p}{y_{\overline{U}}}^p(1 + \abs{V \setminus U}\frac{(1 + (2p - 2)\mu)(1 - 2\alpha)}{p(r + s\mu^p)})
  \end{align*}
  However, $V$ is the result of $1+\eps$-approximate recovery, so
  \begin{align*}
    \norm{p}{y_{\overline{V}}} &\leq \norm{p}{y - \hat{y}} \leq (1 + \eps)\norm{p}{y_{\overline{U}}}\\
    \norm{p}{y_{\overline{V}}}^p &\leq (1 + (2p-1)\eps)\norm{p}{y_{\overline{U}}}^p
  \end{align*}
  for $p \in \{1, 2\}$.  Hence
  \begin{align*}
    \abs{V \setminus U}\frac{(1 + (2p - 2)\mu)(1 - 2\alpha)}{p(r + s\mu^p)} &\leq (2p-1)\eps
  \end{align*}
  for $\alpha \leq 1/4$, this means
  \begin{align*}
    \abs{V \setminus U} &\leq \frac{2\eps(2p-1)p(r + s\mu^p)}{1 + (2p - 2)\mu}.
  \end{align*}
  Plugging in the parameters $p = 1, s =  r =
  \frac{k}{d\eps}, \mu = D$ gives
  \[
  \abs{V \setminus U} \leq \frac{2 \eps ((1 + D^2)r)}{1} = O(\frac{k}{d}).
  \]
  Plugging in the parameters $p = 2, q = 2, r = \frac{k}{d^2\eps^2}, \mu = \frac{1}{d\eps}$ gives
  \[
  \abs{V \setminus U} \leq \frac{12 \eps (3r)}{2 \mu} = \frac{18k}{d}.
  \]
  Hence, for $d = O(c)$, we get the parameters desired in the lemma
  statement, and
  \[
  \abs{V \setminus U} \leq \frac{k}{2c}.
  \]
  Bob can recover $V$ with probability $1-\delta$.  Therefore he can
  output $\hat{x}$ given by $\hat{x}_i = 1$ if $i \in V$ and
  $\hat{x}_i = 0$ otherwise.  This will differ from $x_S$ only within
  $(V \setminus U \cup U \setminus V)$, which is at most $k/c$
  locations.
\end{proof}

\subsection{$k > 1$}

Suppose $p, s, 3r$ satisfy Lemma~\ref{lemma:recoverv} for some
parameter $c$, and let $q = s / k$.  The Gilbert-Varshamov bound
implies that there exists a code $V \subset [q]^r$ with $\log \abs{V}
= \Omega(r \log q)$ and minimum Hamming distance $r/4$.  Let $X
\subset \{0,1\}^{qr}$ be in one-to-one correspondence with $V$: $x \in
X$ corresponds to $v \in V$ when $x_{(a-1)q + b} = 1$ if and only if
$v_{a} = b$.

Let $x$ and $v$ correspond.  Let $S \subset [r]$ with $\abs{S} = k$,
so $S$ corresponds to a set $T \subset [n]$ with $\abs{T} = kq = s$.
Consider arbitrary $\nu$ that satisfies $\norm{p}{\nu} <
\alpha\norm{p}{x}$ and $\norm{\infty}{\nu} \leq \alpha$ for some small
constant $\alpha \leq 1/4$.  We would like to apply
Lemma~\ref{lemma:bitstomeasurements2}, so we just need to show we can
recover $x$ from $A(x + \nu)$ with constant probability.  Let $\nu' =
x_{\overline{T}} + \nu$, so
\begin{align*}
  \norm{p}{\nu'}^p &\leq p(\norm{p}{x_{\overline{T}}}^p + \norm{p}{\nu}^p) \leq p(r - k + \alpha^p r) \leq 3r\\
  \norm{\infty}{\nu'_{\overline{T}}} &\leq 1 + \alpha\\
  \norm{\infty}{\nu'_{T}} &\leq \alpha
\end{align*}
Therefore Lemma~\ref{lemma:recoverv} implies that with probability
$1-\delta$, if Bob is given $A(x_T + \nu') = A(x + \nu)$ he can
recover $\hat{x}$ that agrees with $x_T$ in all but $k/c$ locations.
Hence in all but $k/c$ of the $i \in S$, $x_{\{(i-1)q + 1, \dotsc,
  iq\}} = \hat{x}_{\{(i-1)q + 1, \dotsc, iq\}}$, so he can identify
$v_i$.  Hence Bob can recover an estimate of $v_S$ that is accurate in
$(1-1/c)k$ characters with probability $1-\delta$, so it agrees with
$v_S$ in $(1 - 1/c)(1-\delta)k$ characters in expectation.  If we
apply this in parallel to the sets $S_i = \{k(i-1) + 1, \dotsc, ki\}$
for $i \in [r/k]$, we recover $(1 - 1/c)(1-\delta)r$ characters in
expectation.  Hence with probability at least $1/2$, we recover more
than $(1 - 2(1/c + \delta))r$ characters of $v$.  If we set $\delta$
and $1/c$ to less than $1/32$, this gives that we recover all but
$r/8$ characters of $v$.  Since $V$ has minimum distance $r/4$, this
allows us to recover $v$ (and hence $x$) exactly.  By
Lemma~\ref{lemma:bitstomeasurements2} this gives a lower bound of $m =
\Omega(\log \abs{V}) = \Omega(r \log q)$.  Hence $m =
\Omega(\frac{1}{\eps}k \log \frac{1}{\eps})$ for $\ell_1/\ell_1$
recovery and $m = \Omega(\frac{1}{\eps^2}k)$ for $\ell_2/\ell_2$
recovery.

\subsection{$k = 1, \delta = o(1)$}

To achieve the other half of our lower bounds for sparse outputs, we
restrict to the $k = 1$ case.  A $k$-sparse algorithm implies a
$1$-sparse algorithm by inserting $k-1$ dummy coordinates of value
$\infty$, so this is valid.

Let $p, s, 51r$ satisfy Lemma~\ref{lemma:recoverv} for some $\alpha$
and $D$ to be determined, and let our recovery algorithm have failure
probability $\delta$.  Let $C = 1/(2r\delta)$ and $n = Cr$.  Let $V =
[(s-1)C]^r$ and let $X' \in \{0,1\}^{(s-1)Cr}$ be the corresponding
binary vector.  Let $X = \{0\} \times X'$ be defined by adding $x_0 =
0$ to each vector.

Now, consider arbitrary $x \in X$ and noise $\nu \in \R^{1 + (s-1)Cr}$
with $\norm{p}{\nu} < \alpha\norm{p}{x}$ and $\norm{\infty}{\nu} \leq
\alpha$ for some small constant $\alpha \leq 1/20$.  Let $e^0 / 5$ be
the vector that is $1/5$ at $0$ and $0$ elsewhere.  Consider the sets
$S_i = \{0, (s-1)(i-1)+1, (s-1)(i-1)+2, \dotsc, (s-1)i\}$.  We would
like to apply Lemma~\ref{lemma:recoverv} to recover $(x + \nu + e^0 /
5)_{S_i}$ for each $i$.

To see what it implies, there are two cases: $\norm{1}{x_{sSi}} = 1$
and $\norm{1}{x_{S_i}} = 0$ (since $S_i$ lies entirely in one
character, $\norm{1}{x_{S_i}} \in \{0,1\}$).  In the former case, we
have $\nu' = x_{\overline{S_i}} + \nu + e^0 / 5$ with
\begin{align*}
  \norm{p}{\nu'}^p &\leq (2p-1)(\norm{p}{x_{\overline{S_i}}}^p
+ \norm{p}{\nu}^p + \norm{p}{e^0 / 5}^p) 
%  \ifconf\notag\\&\fi %XXX remove in non-conf
\leq 3(r + \alpha^pr + 1/5^p) < 4r\\
  \norm{\infty}{\nu'_{\overline{S_i}}} &\leq 1 + \alpha\\
  \norm{\infty}{\nu'_{S_i}} &\leq 1/5 + \alpha \leq 1/4
\end{align*}
Hence Lemma~\ref{lemma:recoverv} will, with failure probability
$\delta$, recover $\hat{x}_{S_i}$ that differs from $x_{S_i}$ in at
most $1/c < 1$ positions, so $x_{S_i}$ is correctly recovered.

Now, suppose $\norm{1}{x_{S_i}} = 0$.  Then we observe that
Lemma~\ref{lemma:recoverv} would apply to recovery from $5A(x + \nu +
e^0/5)$, with $\nu' = 5x + 5\nu$ and $x' = e^0$, so
\begin{align*}
  \norm{p}{\nu'}^p &\leq 5^pp(\norm{p}{x}^p + \norm{p}{\nu}^p) \leq 5^pp(r + \alpha^pr) < 51r\\
  \norm{\infty}{\nu'_{\overline{S_i}}} &\leq 5 + 5\alpha\\
  \norm{\infty}{\nu'_{S_i}} &\leq 5\alpha.
\end{align*}
Hence Lemma~\ref{lemma:recoverv} would recover, with failure
probability $\delta$, an $\hat{x}_{S_i}$ with support equal to $\{0\}$.

Now, we observe that the algorithm in Lemma~\ref{lemma:recoverv} is
robust to scaling the input $A(x' + \nu')$ by $5$; the only difference
is that the effective $\mu$ changes by the same factor, which
increases the number of errors $k/c$ by a factor of at most $5$.
Hence if $c > 5$, we can apply the algorithm once and have it work
regardless of whether $\norm{1}{x_{S_i}}$ is $0$ or $1$: if
$\norm{1}{x_{S_i}} = 1$ the result has support $\supp(x_i)$, and if
$\norm{1}{x_{S_i}} = 0$ the result has support $\{0\}$.  Thus we can
recover $x_{S_i}$ exactly with failure probability $\delta$.

If we try this to the $Cr = 1/(2\delta)$ sets $S_i$, we recover all of
$x$ correctly with failure probability at most $1/2$.  Hence
Lemma~\ref{lemma:bitstomeasurements2} implies that $m = \Omega(\log
\abs{X})= \Omega(r\log \frac{s}{r\delta})$.  For $\ell_1/\ell_1$, this
means $m = \Omega(\frac{1}{\eps} \log \frac{1}{\delta})$; for
$\ell_2/\ell_2$, this means $m = \Omega(\frac{1}{\eps^2} \log
\frac{\eps^2}{\delta})$.
\fi

\iffalse
\section{Sparse Upper Bounds}
\begin{theorem}
There is an algorithm using 
$O(\frac{k}{\eps} \log(\frac{n}{k}) + \frac{k}{\eps^2})$
measurements which outputs a $k$-sparse $\hat{x}$ for which
with probability at least $2/3$,
 \[
  \norm{2}{\hat{x}-x} \leq (1 + \eps)\norm{2}{\tail{x}{k}}.
  \]
\end{theorem}
\begin{proof}
Using the fact that the output of the algorithm in \cite{GLPS} is
$O(k)$-sparse, using $O(\frac{k}{\eps} \log(\frac{n}{k}))$
measurements we can first find a $C_1k$-sparse vector $x'$ for which
  \[
  \norm{2}{x'-x} \leq (1 + \eps)\norm{2}{\tail{x}{k}}.
  \]
Here $C_1 > 0$ is a constant. 
Let $S$ be the support of $x'$. In parallel we choose 
$O(\log 1/\eps)$-wise independent hash functions 
$h:[n] \rightarrow [C_2 \eps^{-2}]$, for a constant $C_2 > 0$,
and $g:[n] \rightarrow \{-1,1\}$. We maintain the $C_2\eps^{-2}$-dimensional
vector $v$, where 
$v_i = \sum_{j \mid h(j) = i} g(j) \cdot x_j.$ Observe that for any $j \in [n]$,
if $h(j) = i$, then 
$$v_i \cdot g(j) = x_j + \sum_{j' \mid h(j') = i, j' \neq j} g(j) \cdot g(j') \cdot x_j.$$
Let $z$ denote the restriction of vector $x$ to coordinates $j' \neq j$ for which
$h(j') = i$. 
\begin{fact}\label{fact:khintchine} (Khintchine inequality) (\cite{h82})
For $t \geq 2$, a vector $z$ and a $t$-wise independent random sign
vector $\sigma$ of the same number of dimensions, ${\bf E}[|\langle z,
  \sigma \rangle|^t] \leq \|z\|_2^t (\sqrt{t})^t.$
\end{fact}
Applying Fact \ref{fact:khintchine} with 
\end{proof}
\fi
{\it Acknowledgment:} We thank T.S. Jayram for helpful discussions. 
\ifconf
\bibliographystyle{IEEEtranS}
\else
\bibliographystyle{alpha}
\fi
\bibliography{eps-sparse}

\end{document}